\documentclass{article}
\usepackage[a4paper, total={132mm,190mm}]{geometry}
\usepackage{url}
\usepackage{epsfig}
\usepackage{amsmath}
\usepackage{amsthm}
\usepackage{amssymb}

\usepackage{epsfig}
\usepackage{algorithm}
\usepackage{bm}
\usepackage{alltt}
\usepackage{wasysym}
\usepackage{pifont}
\usepackage{stmaryrd}
\usepackage{color}

\newcommand{\tuple}[1]{\left<{#1}\right>}
\newcommand{\Dom}{\mbox{\sl dom\/}}
\newcommand{\Codom}{\mbox{\sl codom\/}}

\DeclareMathOperator{\maxpref}{\sqcap}
\DeclareMathOperator{\bigmaxpref}{\bigsqcap}

\newtheorem{theorem}{Theorem}
\newtheorem{definition}{Definition}
\newtheorem{lemma}{Lemma}

\newtheorem{remark}{Remark}
\newtheorem{corollary}{Corollary}

\begin{document}
\title{Characterisation of (Sub)sequential Rational Functions over a General Class Monoids}
\author{Stefan Gerdjikov$^{1,2}$ \\~\\
\small
$^1$Faculty of Mathematics and Informatics, \\
\small
Sofia University, Sofia, Bulgaria\\
\small
{\tt stefangerdzhikov@fmi.uni-sofia.bg}\\
\small
$^2$
Institute of Information and Communication Technologies, \\
\small
Bulgarian Academy of Sciences, Sofia, Bulgaria\\
}
\date{}

\maketitle

\begin{abstract}
In this technical report we describe a general class of monoids
for which (sub)sequential rational can be characterised in terms
of a congruence relation in the flavour of Myhill-Nerode relation.
The class of monoids that we consider can be described in terms
of natural algebraic axioms, contains the free monoids, groups,
the tropical monoid, and is closed under Cartesian.  
\end{abstract}
%\keywords{sequential transducers, sequential rational functions, minimisation, congruence relations, monoids}

\section{Inroduction}\label{sec:intro}
The problem to efficiently represent functions $f:\Sigma^*\rightarrow M$ that map
words to some monoid arises in different areas of Natural Language Processing: Speech Recognition, Machine Translation, 
Parsing, Similarity Search. \emph{Finite state transducers} are a natural extension of (classical) finite state automata that provide an efficient 
representation a special class of such functions called \emph{rational functions},~\cite{Eil74,Berstel79,Kempe01,Mo95,Mo97,Mo00,Sakarovitch09}. 

As it is common for most kinds of computational devices, the notion of \emph{determinism} plays an important role since it usually implies more efficient computation. In terms of automata and transducers, the determinism means strongly linear on-line algorithm for parsing the input.
This motivates the interest in deterministic finite state transducers that are called \emph{(sub)sequential transducers}~\cite{Eil74,Sakarovitch09}. 

For (classical) finite state automata it is well known that deterministic automata are equivalent to non-deterministic automata. However, this is not the case for transducers and (sub)sequential transducers~\cite{Ch77,BCPS03,Mo00}. Actually, the latter are capable to represent only a proper class of rational functions called \emph{(sub)sequential rational functions}.

In this paper we consider the characterisation problem of (sub)sequential rational functions. There are two main streams of characterisations known in the literature. The first one characterises the class of (sub)sequential rational functions as rational functions with some additional property, \emph{bounded variation}. This is the kind of characterisation of (sub)sequential rational functions in~\cite{Ch77,Mo00,DRT16,LATA2017}. The second type of characterisation is in terms of congruence relations. This approach bears the flavour of the Myhill-Nerode Theorem,~\cite{HMU01}, for classical finite state automata. Specifically, it departs from an arbitrary function $f:\Sigma^*\rightarrow M$ and defines a congruence relation $\equiv_f$ in terms of the function, but with no regard to its representation. Then the characterisation states more or less: \emph{$\equiv_f$ is of finite index if and only if $f$ is (sub)sequential rational function.} 

Essentially, the first kind of characterisation relates one kind of syntactic representation with another whereas the second kind of characterisation
relates the semantics of the function with its syntactic representation. As such, the first kind of characterisation is useful for practical purposes, whereas the second provides a better theoretical understanding of this class of functions.

In this paper we are considering the second kind of characterisation. Characterisations of the (sub)sequential rational functions $f:\Sigma^*\rightarrow {\cal M}$ in terms of a congruence relation $\equiv_f$ have been studied for different special cases of the monoid ${\cal M}$. The classical result,~\cite{Sakarovitch09}, captures the case where ${\cal M}$ is a free monoid. The characterisation in~\cite{Mo00} deals with the case where ${\cal M}=\tuple{\mathbb{R}^+_0,+,0}$. In~\cite{SouzaMasterThesis} is considered the case of \emph{gcd monoids}. This class captures a wide class of monoids, e.g. groups, free monoids but misses some simple cases like $\tuple{\mathbb{Q}^+_0,+,0}$. In~\cite{Arxive2017} we have shown similar characterisation for yet another class of monoids, \emph{sequentiable structures}. 

In this paper we show a characterisation of (sub)sequential rational functions in terms of congruence relation for functions $f:\Sigma^*\rightarrow {\cal 
M}$ for the class of monoids ${\cal M}$ introduced in~\cite{LATA2018}. This class of monoids is described by five simple algebraic axioms. The only 
additional notion that we need is the relation $a\le_M b$ which is an abbreviation of $b=ac$ for some monoid element $c$. Thus, $\le_M$ is a pre-order on ${\cal M}$. In this framework for each set of monoidal elements, we can consider the set of \emph{lower bounds}, the set of \emph{infimums}, the set of \emph{upper bounds}, and the set of \emph{supremums}, respectively. 

The class of monoids introduced in~\cite{LATA2018} are those that satisfy the following five properties: (i) left cancellation;(ii) right 
cancellation; (iii) any two elements $a,b\in M$ admit an infimum in terms of $\le_M$; (iv) any two elements $a,b\in M$ that have an upper bound in $
{\cal M}$ admit a supremum w.r.t. $\le_M$; (v) if $b\le_M c$ and $b\le_M ac$, then $b\le_M ab$. Groups, free monoids, sequentiable structures, 
tropical monoids ($\tuple{\mathbb{Q}^+_0,+,0}$, $\tuple{\mathbb{R}^+_0,+,0}$, etc.) all satisfy these axioms. 

The gcd monoids can be viewed as monoids satisfying properties (i) and (ii) and additionally every subset of $M$ has a non-empty set of infimums.
In section~\ref{sec:gcd_mge} we shall prove that the gcd monoids also satisfy (iii) and (iv). However, in general they should not respect (v). 

In~\ref{LATA2018} we showed that properties (i)--(v) provide constructive way to minimise any (sub)sequential transducer. We also 
proved that property (v) is essential in order that every regular language over ${\cal M}$ has an infimum. The characterisation that we provide in the current paper is for monoids with properties (i),(ii),(iv), and (v) with additional axiom that we call WLP-axiom. This axiom is satisfied in all the above named monoids, including the gcd monoids. 

The rest of the paper is organised as follows. In Section~\ref{sec:monoids} we recall the basic notions on monoids and formally introduce the relation $\le_M$ along with the terms \emph{infimum, supremum, etc.} that are used throughout the paper. We also recall the definitions of the monoids mentioned above except the gcd monoids that are defined in Setcion~\ref{sec:gcd_mge}. In Section~\ref{sec:automata} we provide the preliminaries on automata and transducers. In Section~\ref{sec:axioms} we formally introduce the properties (i)--(v) and the WLP-axiom. We define the mge monoids and prove some interesting and useful properties about them. We further prove that the classes of monoids considered above are all mge monoids. In Section~\ref{sec:MNR} we define the congruence $\equiv_f$, state and prove our characterisation result. In Section~\ref{Example} we discuss the necessity of the WLP-axiom. We prove that a non-uniform version of this axiom is
necessary for the characterisation we strive at under natural assumptions for the monoid. In Section~\ref{sec:gcd_mge} we recall the definition of the gcd monoids and compare them against the mge monoids. We conclude in Section~\ref{sec:conclusion}.

\section{Monoids}\label{sec:monoids}
We open this section with the definition of a monoid,~\cite{Eil74,Sakarovitch09}. 
In Subsection~\ref{subsec:relations} we consider some useful relations on monoids that play an important role throughout the paper.
In Subsection~\ref{subsec:classes} we provide some examples for monoids. A reader familiar with the basic notions may prefer to look
only at Subsection~\ref{subsec:relations}.
\begin{definition}
A monoid is a structure ${\cal M}=\tuple{M,\circ,e}$ where:
\begin{enumerate}
\item $M$ is a set, the support of ${\cal M}$,
\item $\circ:M\times M\rightarrow M$ is an associative operation, i.e.:
\begin{equation*}
\forall a,b,c\in M( a\circ (b\circ c) = (a\circ b)\circ c)
\end{equation*}
\item $e\in M$ is an unit element w.r.t. $\circ$, i.e.:
\begin{equation*}
\forall a\in M(a=e\circ a=a\circ e).
\end{equation*} 
\end{enumerate}
\end{definition}
Given a monoid ${\cal M}$, we can canonically lift the product in ${\cal M}$ to products of subsets of $M$.
\begin{definition}\label{regular_operations}
Let ${\cal M}$ be a monoid, for subsets $A,B\subseteq M$ we define:
\begin{equation*}
AB = A\circ B =\{a\circ b\,|\, a\in A,b\in B\}.
\end{equation*}
For a natural number $n\in \mathbb{N}$ we define:
\begin{equation*}
A^n =\begin{cases} \{e\} \text{ if } n=0\\
A\circ A^{n-1} \text{ if } n>0.
\end{cases}
\end{equation*}
Finally, an iteration of a subset $A\subset M$ is:
\begin{equation*}
A^* = \bigcup_{n=0}^\infty A^n.
\end{equation*}
\end{definition}
\begin{remark}
For better readability, for an element $m\in M$ and a set $A\subseteq M$ we shall write:
\begin{equation*}
mA=m\circ A \text{ and } Am=A\circ m
\end{equation*}
as abbreviation for:
\begin{equation*}
\{m\}\circ A\text{ and } A\circ \{m\},
\end{equation*} 
respectively.
\end{remark}

\subsection{Relations on Monoids}\label{subsec:relations}
\begin{definition}\label{def:preorder}
For a monoid ${\cal M}$ and elements $a,b\in {\cal M}$, we say that $a$ is less than or equal to $b$ and write $a\le_M b$ if and only if there is an element $c$ with $ac=b$.
\end{definition}
The relation $\le_M$ is clearly transitive and reflexive. Thus, it defines a pre-order on $M$. Therefore we can decompose $\le_M$ into an equivalence relation and a partial order in a canonical way: 
\begin{definition}
Let ${\cal M}$ be a monoid. The relation $\sim_M$ is defined as:
\begin{equation*}
a\sim_M b \iff a\le_M b \text{ and } b\le_M a.
\end{equation*}

\end{definition}
\begin{lemma}\label{lemma:factor}
Let $\cal M$ be a monoid. Then $\sim_M$ is an equivalence relation. 
\end{lemma}
\begin{proof}
Immediate. 
\end{proof}

\begin{lemma}\label{lemma:order}
Given a monoid ${\cal M}$, the relation on its factor ${\cal M}/_{\sim_M}$:
\begin{equation*}
[ a ]_{\sim M} \le [b]_{\sim M} \iff a\le_M b
\end{equation*}
is well-defined and is a partial order on ${\cal M}/_{\sim_M}$.
\end{lemma}
\begin{proof}
If $a'\sim_M a$, $a\le_M b$, and $b\sim_M b'$, then $a'\le_M a$ and $b\le_M b'$. By the transitivity of $\le_M$, we get that
$a'\le_M b'$. Therefore the relation $\le$ is well-defined on ${\cal M}/_{\sim_M}$. The same reasoning shows that $\le$ is transitive,
and it is obvious that it is reflexive. To prove that $\le$ is antisymmetric consider $[a]\le [b]$ and $[b]\le [a]$. Hence, $a\le_M b$ and
$b\le_M a$. Therefore, $a\sim_M b$ and consequently $[a]=[b]$. 
\end{proof}

\begin{definition}\label{infimum}
Let $S\subseteq {\cal M}$. A lower bound for $S$ is any element $a\in {\cal M}$ such that $a\le_M s$ for all $s\in S$.
Similarly, an upper bound for $S$ is any element $b\in {\cal M}$ such that $s\le_M b$ for all $s\in S$. We denote with
$low(S)$ and $up(S)$ the set of lower and upper bounds for $S$, respectively, i.e.:
\begin{eqnarray*}
low(S) &=& \{a\in M \,|\, \forall s\in S(a\le_M s)\}\\
up(S) & = & \{b\in M \,|\, \forall s\in S(s\le_M b)\}. 
\end{eqnarray*}
We define the sets of infimums and supremums for $S$ as:
\begin{eqnarray*}
\inf S &=&\{ l \in low(S) \,|\, \forall m\in low(S) (m\le_M l)\} \\
\sup S &=& \{u\in up(S) \,|\, \forall m\in up(S) (u\le_M m)\}.
\end{eqnarray*}

An infimum for $S$ is any element $i\in \inf S$. Similarly, a supremum for $S$ is any element $m\in \sup S$.
\end{definition}
\begin{remark}\label{rem:supremum}
Since $s\le_M b$ is equivalent to $b\in sM$ we can express $up(S)$ as:
\begin{equation*}
up(S) = \bigcap_{s\in S} sM.
\end{equation*}
In particular, if $u\in up(S)$, then $uM \subseteq \bigcap_{s\in S} sM$. Furthermore, in the special case where $u\in \sup(S)$,
we have that for any $m\in \bigcap_{s\in S} sM$ it is the case that $u\le_M m$, i.e. $m\in uM$. With this remark it is easy to see that $u\in \sup(S)$ is equivalent to:
\begin{equation*}
\bigcap_{s\in S} sM =uM.
\end{equation*}
\end{remark}
\begin{remark}\label{rem:infimum}
In general a set $S$, may have or may have no infimums. Even, if $S$ has an infimum $i\in {\cal M}$, it should not be unique. Actually, in this case the set of infimums of $S$ is $\inf S =[i]_{\sim_M}$.
\end{remark}
\begin{definition}
Given a monoid ${\cal M}=\tuple{M,\circ,e}$ we say that an element $c\in M$ is invertible iff there exists $c'\in M$ with:
\begin{equation*}
cc'=e=c'c.
\end{equation*}
\end{definition}

So far we have been concerned with the multiplication on the right hand side. In certain situations we will need to consider also multiplications on the left. However they will concern only invertible elements. For these purposes we give the following definition:
\begin{definition}\label{left_congruence}
Let ${\cal M}$ be a monoid. For an integer number $n\ge 1$ we define the relation $\equiv^{(n)}_{M}\subseteq M^n\times M^n$ as:
\begin{equation*}
{\bf a}\equiv^{(n)}_M {\bf b} \iff \exists u\in M( \forall i \le n(u{\bf a}_i={\bf b}_i) \text{ and } u \text{ is invertible}).
\end{equation*}
\end{definition}
\begin{lemma}\label{lemma:left_congruence}
The relation $\equiv^{(n)}_M$ is an equivalence relation.
\end{lemma}
\begin{proof}
Since $e$ is invertible, $\equiv^{(n)}_M$ is reflexive. If $a=ub$ for some invertible element $u\in  M$,
then its inverse, $u'\in M$, is also invertible and further $u' a = b$. Hence $\equiv_M$ is symmetric.
Finally, if $a=ub$ and $b=vc$ for some invertible elements $u,v\in M$, then clearly $a=uvc$.
Finally, if $u'$ and $v'$ are the inverse of $u$ and $v$, then $uvv'u'=e$ and $v'u'uv=e$, showing that $uv$ is invertible.
Hence $\equiv^{(n)}_M$ is also transitive.
\end{proof}

\subsection{Classes of Monoids}\label{subsec:classes}
\begin{definition}\label{def:free}
Given a set $\Sigma$, the free monoid generated by $\Sigma$ is defined as $\tuple{\Sigma^*,\circ,\varepsilon}$ where:
\begin{enumerate}
\item $\Sigma^*=\{\tuple{a_1,a_2,\dots,a_n}\,|\, n\in \mathbb{N}, a_i\in \Sigma\}$ is the set of all finite sequences of elements in $\Sigma$.
\item $\tuple{a_1,a_2,\dots,a_n} \circ\tuple{b_1,b_2,\dots,b_m}=\tuple{a_1,a_2,\dots,a_n,b_1,b_2,\dots,b_m}$, i.e. $\circ$ is the concatenation of sequences,
\item $\varepsilon=\tuple{}$ is the empty sequence.
\end{enumerate}
\end{definition}

\begin{definition}\label{def:tropical}
We refer to the structure $\tuple{\mathbb{R}_0^+,+,0}$ as the tropical monoid.
\end{definition}
\begin{remark}
It is apparent that the tropical monoid is a monoid.
\end{remark}

\begin{definition}\label{def:group}
A group is a monoid, ${\cal{M}}$, all whose elements are invertible.
\end{definition}
\iffalse
\begin{definition}\label{def:inf_group}
An infinitary group is a group ${\cal {G}}=\tuple{G,\circ,e}$ with the property that for any $a,b,c\in G$ such that:
\begin{equation*}
\{a^nbc^n\,|\, n\in \mathbb{N}\} \text{ is finite} \Rightarrow abc=b.
\end{equation*}
\end{definition}
\fi
Sequentiable structures were defined in~\cite{LATA2017}.
\begin{definition}\label{def:seq_str}
A sequentiable structure is a monoid, ${\cal{M}}=\tuple{M,\circ,e}$, s.t.:
\begin{enumerate}
\item for every $a,b\in M$, $|\inf(\{a,b\})|=1$,
\item there is homomorphism $\|.\|:{\cal {M}}\rightarrow\mathbb{R}^+_0$ with trivial kernel, i.e.:
\begin{eqnarray*}
\forall a,b\in M(\|a\circ b\| = \|a\| + \|b\|) \text{ and } \|a\|=0 \iff a=e.
\end{eqnarray*}
\item for all $a,b,c\in {\cal M}$ if $a\le_M bc$ and $\|a\|\le \|b\|$, then $a\le_M b$.
\end{enumerate} 
\end{definition}
\begin{remark}
It is easy to see that the free monoids represent a subclass of sequentiable structures. The tropical monoid, $\tuple{\mathbb{R}_0^+,+,0}$ is also an instance of a sequentiable structure.
\end{remark}

\begin{definition}\label{def:cartesian_product}
Let ${\cal M}_i=\tuple{M_i,\circ_i,e_i}$ for $i=1,2$ be monoids. The Cartesian Product of ${\cal{M}}_1$ and ${\cal{M}}_2$ is ${\cal{M}}={\cal{M}}_1\times{\cal{M}}_2$ where ${\cal{M}}=\tuple{M_1\times M_2,\circ,\tuple{e_1,e_2}}$ s.t.:
\begin{equation*}
\tuple{a_1,a_2}\circ\tuple{b_1,b_2}=\tuple{a_1\circ_1b_1,a_2\circ_2b_2}.
\end{equation*}
\end{definition}
\begin{remark}
It is easy to see that Cartesian Product of monoids is a monoid.
\end{remark}

\section{Finite State Transducers and Automata}\label{sec:automata}
This section is preliminary on automata and transducers,~\cite{Eil74,Sakarovitch09}. A reader familiar with the basic notions on automata
can skip this section.
\subsection{Automata}
\begin{definition}\label{def:fsa}
Given a monoid ${\cal M}$, a finite state monoidal automaton over ${\cal M}$
is ${\cal{A}}=\tuple{{\cal{M}},Q,I,F,\Delta,\iota,\Psi}$ where:
\begin{enumerate}
\item $Q$ is a finite set (of states),
\item $I\subseteq Q$ is a set (of initial states),
\item $F\subseteq Q$ is a set (of final states),
\item $\Delta\subseteq Q\times{\cal{M}}\times{Q}$ is a finite set (of transitions),
\item $\iota:I \rightarrow {\cal M}$ is an initial output function,
\item $\Psi: F\rightarrow {\cal M}$ is a final output function.
\end{enumerate}
\end{definition}

\begin{definition}\label{def:path}
Given a monoidal automaton, ${\cal{A}}=\tuple{{\cal{M}},Q,I,F,\Delta,\iota,\Psi}$, a non-trivial path is a 
non-empty sequence of transitions, $\pi=\tuple{p_0,m_1,p_1}\dots\tuple{p_{n-1},m_n,p_n}$ with $\tuple{p_i,m_{i+1},p_{i+1}}\in\Delta$. The source of $\pi$ is $\sigma(\pi)=p_0$, the target of $\pi$ is $\tau(\pi)=p_n$, the length of $\pi$ is $|\pi|=n$, and the label of $\pi$ is:
\begin{equation*}
\ell(\pi) =m_1\circ m_2\circ\dots\circ m_n.
\end{equation*}

For each state $p\in Q$ we have a void path $\pi_p=(p)$ with no transitions and $\sigma(\pi_p)=\tau(\pi_p)=p$, $|\pi_p|=0$, and $\ell(\pi_p)=e$.

A path in ${\cal A}$ is either a non-trivial or a void path. A path $\pi$ is called successful if $\sigma(\pi)\in I$ and $\tau(\pi)\in F$.
\end{definition}

\begin{definition}\label{def:gen_transitions}
For a monoidal automaton, ${\cal{A}}=\tuple{{\cal{M}},Q,I,F,\Delta,\iota,\Psi}$, we define the generalised transitions of length $n$ as:
\begin{equation*}
\Delta^n = \{\tuple{\sigma(\pi),\ell(\pi),\tau(\pi)}\,|\, \pi\text{ is a path in }{\cal{A}} \text{ with }|\pi|=n\}
\end{equation*}
A generalised transition in ${\cal{A}}$ is any element of $\Delta^*=\bigcup_{n=0}^{\infty}\Delta^n$.
We also define $\Delta^{<n}=\bigcup_{k=0}^{n-1} \Delta^k$ and $\Delta^{\le n}=\bigcup_{k=0}^{n} \Delta^k$.
\end{definition}

\begin{definition}
A language recognised by a finite state monoidal automaton, ${\cal{A}}=\tuple{{\cal{M}},Q,I,F,\Delta,\iota,\Psi}$, is:
\begin{equation*}
{\cal{L}}({\cal{A}}) =\{ \iota(s)\circ m \Psi(f) \in {\cal{M}}\,|\,\exists s\in I\exists f\in F(\tuple{s,m,f}\in\Delta^*)\}
\end{equation*}
Two monoidal automata (over the same monoid) are said to be equivalent, if they recognise the same language.
\end{definition}

\begin{definition}
Given a monoidal finite state automaton, ${\cal{A}}=\tuple{{\cal{M}},Q,I,F,\Delta,\iota,\Psi}$, and a state $q\in Q$ the language of $q$ w.r.t. ${\cal{A}}$ is:
\begin{eqnarray*}
{\cal{L}}_{\cal{A}}(q) &=&{\cal{L}}({\cal{A}}_q), \text{ where }
{\cal{A}}_q=\tuple{{\cal{M}},Q,\{q\},F,\Delta,e_q,\Psi}.
\end{eqnarray*}
Here $e_q:\{q\}\rightarrow {\cal M}$ is the function defined by $e_q(q)=e$.
\end{definition}

\begin{definition}
Given a monoidal finite state automaton, ${\cal{A}}=\tuple{{\cal{M}},Q,I,F,\Delta,\iota,\Psi}$, we say that a state $q\in Q$ is accessible (co-accessible) iff there is a path $\pi$ in ${\cal{A}}$ with $\sigma(\pi)\in I$ and $\tau(\pi)=q$ ($\sigma(\pi)=q$ and $\tau(\pi)\in F$, respectively). The automaton ${\cal{A}}$ is called trimmed if all states $q\in Q$ are both accessible and co-accessible.
\end{definition}
\begin{remark}\label{remark:trimmed}
For every monoidal finite state automaton there is an equivalent trimmed monoidal finite state automaton.
\end{remark}

\begin{definition}\label{def:reg_languages}
Given a monoid ${\cal M}$ the set of regular languages over ${\cal M}$ is the inclusion-wise least set $Reg({\cal{M}})$ such that:
\begin{enumerate}
\item $\emptyset\in Reg({\cal{M}})$,
\item if $m\in M$, then $\{m\}\in Reg({\cal{M}})$,
\item if $L_1,L_2\in Reg({\cal{M}})$, then:
$$L_1\cup L_2\in Reg({\cal{M}}),\quad L_1\circ L_2\in Reg({\cal{M}}), \text{ and } L_1^*\in Reg({\cal{M}}).$$
\end{enumerate}
\end{definition}

The Kleene's Theorem states that:
\begin{theorem}
For any monoid ${\cal{M}}$ a set $S\subseteq {\cal{M}}$ is regular language over ${\cal{M}}$ if and only if there is a finite state monoidal automaton ${\cal{A}}$ over ${\cal{M}}$ with ${\cal{L}}({\cal{A}})=S$.
\end{theorem}

\subsection{Transducers}
\begin{definition}
An alphabet is any finite set $\Sigma$. A word is any element of the free monoid $\Sigma^*$. For a set $S\subseteq \Sigma^*$ and a word $\alpha\in \Sigma^*$ we define:
\begin{equation*}
\alpha^{-1} S = \{\beta\,|\, \alpha\beta\in S\}.
\end{equation*}
\end{definition}
\begin{definition}
Given an alphabet $\Sigma$ and a monoid ${\cal {M}}$, a $\Sigma-{\cal{M}}$-transducer is any finite state automaton
${\cal{T}}=\tuple{\Sigma^*\times{\cal{M}},Q,I,F,\Delta,\iota,\Psi}$.
\end{definition}
\begin{remark}
It is easy to see that for every $\Sigma-{\cal{M}}$-transducer ${\cal{T}}=\tuple{\Sigma^*\times{\cal{M}},Q,I,F,\Delta,\iota,\Psi}$ there is an equivalent $\Sigma-{\cal{M}}$-transducer ${\cal{T}}'=\tuple{\Sigma^*\times{\cal{M}},Q',I',F',\Delta',\iota',\Psi'}$ with the following three additional properties:
\begin{enumerate}
\item $\Codom(\iota')\subseteq \{\varepsilon\}\times {\cal M}$,
\item $\Codom(\Psi')\subseteq \{\varepsilon\}\times {\cal M}$,
\item $\Delta'\subseteq Q\times ((\Sigma\cup\{\varepsilon\})\times{\cal{M}})\times{Q}$.
\end{enumerate}
This kind of transducers are known as one-letter transducers.
In the sequel we will be considering only one-letter transducers. To stress their main properties we shall write:
$${\cal{T}}'=\tuple{\Sigma\times{\cal{M}},Q',I',F',\Delta',\iota',\Psi'}$$
omitting the star in $\Sigma^*$.
\end{remark}

\begin{definition}
Let ${\cal{T}}=\tuple{\Sigma\times{\cal{M}},Q,I,F,\Delta,\iota,\Psi}$ be a $\Sigma-{\cal{M}}$-transducer. For a path $\pi$ in ${\cal{T}}$ we define $input(\pi)$ and $out(\pi)$ such that $\ell(\pi)=\tuple{input(\pi),out(\pi)}\in\Sigma^*\times{\cal{M}}$.
\end{definition}

\iffalse
\begin{definition}
Let ${\cal{T}}=\tuple{\Sigma\times{\cal{M}},Q,I,F,\Delta}$. For a state $q\in Q$ we define the domain, the range, and the set of longest common beginnings for $q$, respectively, as follows:
\begin{eqnarray*}
{\cal{D}}_{\cal{T}}(q) & = & \{\alpha\in \Sigma^* \,|\, \exists m\in {\cal {M}}\tuple{\alpha,m}\in {\cal{L}}_{\cal{T}}(q)\}\\
{\cal{R}}_{\cal{T}}(q) & = & \{m\in {\cal{M}} \,|\, \exists\alpha\in \Sigma^*\tuple{\alpha,m}\in {\cal{L}}_{\cal{T}}(q)\} \\
\end{eqnarray*}
\end{definition}
\fi
\begin{definition}
A $\Sigma-{\cal{M}}$-transducer, ${\cal{T}}$, is called functional if the language ${\cal{L}}({\cal{T}})$ is a graph of a (partial) function
${\cal{O}}_{\cal{T}}:\Sigma^*\rightarrow{\cal{M}}$. 

Similarly, if ${\cal L}_{\cal{T}}(q)$ is a graph of a partial function from $\Sigma^*\rightarrow {\cal M}$ we shall denote it with ${\cal{O}}^{(q)}_{\cal{T}}:\Sigma^*\rightarrow {\cal M}$.
\end{definition}

\begin{definition}
A $\Sigma-{\cal{M}}$-transducer, ${\cal{T}}$, is called onward if for every state $q$ of ${\cal{T}}$ it holds that $e\in \inf {\cal{R}}_{\cal{T}}(q)$.
\end{definition}

\subsection{(Sub)sequential Transducers}
\begin{definition}\label{def:subseq_trans}
Given an alphabet $\Sigma$ and a monoid ${\cal{M}}$ a subsequential transducer is
a one-letter transducer ${\cal{T}}=\tuple{\Sigma,{\cal{M}},Q,I,F,\Delta,\iota,\Psi}$ where:
\begin{enumerate}
\item $|I|=1$,
\item $\Delta\subseteq Q\times (\Sigma\times {\cal M})\times Q$ is a graph of a function $Q\times \Sigma\rightarrow M\times Q$.
\end{enumerate}
To emphasise the components of a (sub)sequential transducer we shall use the notation:
\begin{equation*}
{\cal{T}}=\tuple{\Sigma,{\cal{M}},Q,s,F,\delta,\lambda,\iota,\Psi}
\end{equation*}
where $I=\{s\}$ and $\delta:Q\times \Sigma\rightarrow Q$ and $\lambda:Q\times \Sigma\rightarrow M$ are partial functions with:
\begin{equation*}
\Delta = \{\tuple{p,\tuple{a,\lambda(p,a)},\delta(p,a)} \,|\, p\in Q, a\in \Sigma\}
\end{equation*}
We call $\delta$ transition function and $\lambda$ -- output function. We all tacitly identify $\iota:\{s\}\rightarrow M$ with $\iota(s)$.
\end{definition}
\begin{definition}
A subsequential transducer is called complete if its transition function is total.
\end{definition}
\begin{remark}
For a (sub)sequential transducer the transition relation $\Delta$ is a graph of a function mapping $Q\times\Sigma\rightarrow M\times Q$. This means that given a state $q$ and a word $\alpha\in \Sigma^*$ there is at most one path $\pi$ with origin $p$ and $input(\pi)=\alpha$. This implies that $\Delta^*$ is also a graph of a function mapping $Q\times\Sigma^*\rightarrow M\times Q$.
The functions $\delta^*:Q\times\Sigma^*\rightarrow Q$ and $\lambda^*:Q\times \Sigma^*\rightarrow M$ that we formally define below represent its projections, respectively.
\end{remark}
\begin{definition}\label{def:subseq_trans}
Given a subsequential transducer
${\cal{T}}=\tuple{\Sigma,{\cal{M}},Q,I,F,\Delta,\iota,\Psi}$ we define $\delta^*:Q\times \Sigma^*\rightarrow Q$ and $\lambda^*:Q\times \Sigma^*\rightarrow{\cal{M}}$ as follows:
\begin{eqnarray*}
\delta^*(p,\alpha) = \begin{cases} q \text{ if } \exists m\in M(\tuple{p,\tuple{\alpha,m},q}\in\Delta^*) \\
\text{not defined, else.}
\end{cases}\\
\lambda^*(q,\alpha) = \begin{cases}m \text{ if } \exists q\in Q(\tuple{p,\tuple{\alpha,m},q}\in\Delta^*) \\
\text{not defined, else.}
\end{cases}
\end{eqnarray*}
\end{definition}
\begin{remark}
Note that every (sub)sequential transducer is functional. However, the converse should not be true.
\end{remark}

\begin{definition}
For a subsequential transducer ${\cal{T}}=\tuple{\Sigma,{\cal{M}},Q,s,F,\delta,\lambda,\iota,\Psi}$, we define a (partial) function ${\cal{O}}_{\cal{T}}:\Sigma^*\rightarrow {\cal{M}}$ is the function represented by the subsequential transducer ${\cal T}$.
\end{definition}
\iffalse
\begin{definition}
For a subsequential transducer ${\cal{T}}=\tuple{\Sigma,{\cal{M}},Q,s,F,\delta,\lambda,\iota,\Psi}$ and a state $q\in Q$ we define
the (sub)sequential transducer ${\cal{T}}(q)$ as:
\begin{equation*}
{\cal{T}}(q)=\tuple{\Sigma,{\cal{M}},Q,q,F,\delta,\lambda,e,\Psi}
\end{equation*}
and we define the function, domain, and range of a state $q$ as:
\begin{equation*}
{\cal{O}}_{\cal{T}}^{(q)} = {\cal{O}}_{{\cal{T}}(q)}, \quad
{\cal{D}}_{\cal{T}}(q) =  \Dom({\cal{O}}_{\cal{T}}^{(q)}), \text{ and }
{\cal{R}}_{\cal{T}}(q) = \Codom({\cal{O}}_{\cal{T}}^{(q)}),
\end{equation*}
respectively.
\end{definition}
\fi

\section{Axioms}\label{sec:axioms}
In this section is fundamental for the understanding of the results in subsequent sections. First we revise the
definitions of mge monoids, the GCLF- and LSL-axioms that were introduced in~\cite{LATA2018}. We complete
these definitions by shedding additional light on the properties such monoids possess. At the end of Subsection~\ref{subsec:gclf},
we list the main results from~\cite{LATA2018}. This can be considered as a motivation to consider mge monoids with GCLF- and LSL-axioms
and strive at characterisation of the (sub)sequential rational functions with range in such monoids.
In Subsection~\ref{subsec:lpa} we consider one more axiom. It is a not natural, second order formula, that, at first glance seems artificial.
Yet, as we shall see in Section~\ref{Example} an axiom with such flavour is necessary condition for the characterisation we are looking for. 

Throughout this section we also show that all the axioms introduced here are valid for free monoids, tropical monoids, sequentiable structures, and groups. We also prove that they are closed under Cartesian Product of monoids.
\subsection{MGE Axioms}\label{subsec:mge}
\begin{definition}
We say that a monoid ${\cal M}$ satisfies the Left Cancellation Axiom (LC-axiom), if:
\begin{equation*}
\forall a,b,c\in {\cal M}(cb=ca\Rightarrow a=b).
\end{equation*}
In this case, if $a\le_M b$ we shall denote with $\frac{b}{a}$ the unique element such that:
\begin{equation*}
a\circ \frac{b}{a}=b.
\end{equation*}
\end{definition}
\begin{definition}\label{RC}
We say that a monoid ${\cal M}$ satisfies the Right Cancellation Axiom (RC-axiom) if:
\begin{equation*}
\forall a,b,c\in {\cal M} (ac= bc \Rightarrow a=b).
\end{equation*}
%In this case, for element $d$ and $c$, we shall use $d-c\in {\cal M}\cup\{\bot\}$ to denote the unique element 
%$a$ such that $ac=d$ if such exists. By default, $d-c=\bot$ if no element $a$ with $ac=d$ exists.
\end{definition}
\begin{definition}\label{RMGE}
We say that a monoid ${\cal M}$ satisfies the Right Most General Equaliser Axiom (RMGE-axiom) if:
\begin{equation*}
\forall a,b\in {\cal M} (up(\{a,b\})\neq\emptyset\Rightarrow \sup(\{a,b\})\neq\emptyset).
\end{equation*}
If this is the case and $up(\{a,b\})\neq \emptyset$ we shall write $a\vee b$ to denote some arbitrary but fixed
witness $a\vee b\in \sup(\{a,b\})$. We shall assume that $a\vee b$ is undefined if $up(\{a,b\})=\emptyset$.
Finally, given a finite sequence of elements, $\{a_i\}_{i=1}^n$ we shall write $\bigvee_{i=1}^na_i$ as an abbreviation for:
\begin{equation*}
\bigvee_{i=1}^na_i = (\dots((a_1\vee a_2)\vee a_3)\dots \vee a_{n-1})\vee a_n.
\end{equation*}
\end{definition}

In the sequel we describe some simple consequences of the above axioms and revisit the notion
of an mge monoid that was introduced in a previous work,~\cite{LATA2018}.
\iffalse
\begin{lemma}\label{lemma:factor_LC}
If a monoid ${\cal M}$ satisfies the LC-axiom, then for any $a\in M$:
\begin{equation*}
[a]_{\sim_M} = a {\cal I},
\end{equation*}
where ${\cal I}$ is the set  of invertible elements of ${\cal M}$.
\end{lemma}
\begin{proof}
Let $a\sim_M b$. Then $a\le_M b$ and $b\le_M a$. By the LC-axiom, it follows that $\frac{a}{b}$ and $\frac{b}{a}$ are defined.
Hence $a\circ\frac{b}{a}\circ\frac{a}{b}=b\circ \frac{a}{b}=a$ and by the LC-axiom, we get $\frac{a}{b}\frac{b}{a}=e$.
Analogously, $\frac{b}{a}\frac{a}{b}=e$ and hence $\frac{b}{a}$ is invertible,
proving that $b\in a{\cal I}$. 
Conversely, if $b\in a{\cal I}$, then $b=ac$ for some invertible element. Let $c'$ be the inverse of $c$, i.e.
$cc'=c'c=e$. Hence $a=ae=acc'=bc'$. Therefore $b\le_M a$. The equation $b=ac$ shows that $a\le_M b$. Hence $a\sim_M b$.
\end{proof}
\fi
\begin{definition}\label{mge:monoid}
A monoid ${\cal M}$ is called an mge-monoid if it satisfies RMGE-, RC-, and LC-axioms.
\end{definition}
\begin{lemma}\label{lemma:mge_equivalent}
Let ${\cal M}$ be an mge monoid and $a,b\in {\cal M}$ be such that $up(\{a,b\})\neq \emptyset$.
Then there are elements $c,d \in {\cal M}$ such that:
\begin{enumerate}
\item $ac=bd$,
\item if $ax=by$ for some $x,y\in {\cal M}$, then $c\le_M x$ and $d\le_M y$ and $\frac{x}{c}=\frac{y}{d}$.
\end{enumerate} 
\end{lemma}
\begin{proof}
Since ${\cal M}$ is an mge monoid, we have that $a\vee b\in \sup(\{a,b\})$. Hence
$a\le_M a\vee b$ and $b\le_M a\vee b$. Therefore $c=\frac{a\vee b}{a}$ and $d=\frac{a\vee b }{b}$ are well-defined and
$ac=bd=a\vee b$. We prove that $c$ and $d$, as defined, satisfy also the second property. To this end, let $ax=by$. 
Then $a\le_M ax$ and $b\le_M by$.
Hence, $a\vee b \le_M ax=by$. Therefore, $a\vee b=ac\le_M ax$ and by LC-axiom, we get that $c\le_M x$.
Hence we can write $ax=ac\frac{x}{c}=(a\vee b)\frac{x}{c}$. Similar argument shows that $d\le_M y$ and $by=(a\vee b)\frac{y}{d}$. Therefore by the LC-axiom, since $(a\vee b)\frac{y}{d}=by=ax=(a\vee b)\frac{x}{c}$, we conclude that $\frac{x}{c}=\frac{y}{d}$.
\end{proof}
\begin{definition}
An equaliser for a tuple $\tuple{a_1,a_2,\dots,a_n}\in {\cal M}^n$ is a tuple \\$\tuple{u_1,u_2,\dots,u_n}\in {\cal M}^n$ such that:
\begin{equation*}
\forall i,j (a_iu_i=a_ju_j).
\end{equation*}
A most general equaliser (mge) for $\tuple{a_1,a_2,\dots,a_n}$ is a tuple \\$\tuple{m_1,m_2,\dots,m_n}\in {\cal M}^n$ such that:
\begin{enumerate}
\item $\tuple{m_1,\dots,m_n}$ is an equaliser for $\tuple{a_1,\dots,a_n}$,
\item for any equaliser $\tuple{u_1,\dots,u_n}$ for $\tuple{a_1,\dots,a_n}$ there is an element $d\in {\cal M}$ such that:
\begin{equation*}
\forall i (u_i = m_i d).
\end{equation*}
\end{enumerate}
\end{definition}
\iffalse
\begin{remark}
In our previous work,~\cite{TCS-submitted} we have defined the mge monoids as monoids ${\cal M}$ that satisfy:
\begin{enumerate}
\item for all $a,b\in {\cal M}$ if $\tuple{a,b}$ has an equaliser, then it has a most general equaliser, 
\item the RC-axiom.
\end{enumerate}
There we proved that each such monoid obeys also the LC-axiom.
Since $up(\{a,b\})\neq \emptyset$ is equivalent to the premise that $\tuple{a,b}$ is equalisable, Lemma~\ref{lemma:mge_equivalent}
now states that the current Definition~\ref{mge:monoid} is equivalent to the original definition of an mge monoid.
\end{remark}
\fi
We restate two further results about mge monoids that we shall use:
\begin{lemma}\label{lemma:mge}
Let ${\cal M}$ be an mge monoid.
Then for any $a_1,a_2,\dots,a_n\in  M$ the tuple $\tuple{a_1,\dots,a_n}$ is equalisable iff:
\begin{equation*}
\bigvee_{i=1}^n a_i =(\dots ((a_1\vee a_2)\vee a_3) \dots )\vee a_n \text{ is defined}
\end{equation*}
and in this case $\tuple{\frac{\bigvee a_i}{a_1},\frac{\bigvee a_i}{a_2},\dots,\frac{\bigvee a_i}{a_n}}$ is an mge for $\tuple{a_1,\dots,a_n}$.
\end{lemma}
\begin{proof}
By definition, $a\vee b$ is defined iff $aM\cap bM\neq \emptyset$ and in this case, by Remark~\ref{rem:supremum}:
\begin{equation*}
aM\cap bM=(a\vee b)M.
\end{equation*}
First assume that $a=\bigvee_{i=1}^n a_i$ is defined. Then:
\begin{equation*}
a M = \bigcap_{i=1}^n a_i M 
\end{equation*}
and therefore $a\in \bigcap_{i=1}^n a_i M$. Hence, there are elements $u_i\in M$ with $a_iu_i=a$ showing that
 $\tuple{a_1,\dots,a_n}$ is equalisable.

Conversely, assume that $\tuple{a_1,\dots,a_n}$ is equalisable. Let $\tuple{u_1,\dots,u_n}$ be an equaliser for $\tuple{a_1,\dots,a_n}$. Thus, there is an $m$ with $m=a_iu_i \in a_iM$ for each $i$, witnessing that $\bigcap_{i=1}^n a_iM\neq\emptyset$. By above, we conclude that:
\begin{equation*}
\bigvee_{i=1}^n a_i \text{ is defined and } \left(\bigvee_{i=1}^n a_i\right) M =\bigcap_{i=1}^n a_i M.
\end{equation*}
Let $a=\bigvee_{i=1}^n a_i$. Since $a\in aM\subseteq a_jM$ we get that $a_j\le_M a$ for $j\le n$. Therefore $\frac{a}{a_j}$ is defined. Since
$m\in aM$ we have $a\le_M m$ and hence $\frac{m}{a}$ is defined. Therefore for each $j\le n$ we have:
\begin{equation*}
a_ju_j=m=a \frac{m}{a}= a_j \frac{a}{a_j} \frac{m}{a}
\end{equation*}
and by the LC-axiom we deduce $u_j=\frac{a}{a_j} \frac{m}{a}$. Since $a=\bigvee_{i=1}^n a_i$, the result follows.
\end{proof}

\begin{lemma}\label{lemma:mge_partition}
Let ${\cal{M}}$ be an mge monoid. For any $a,b,c,d\in M$ it holds that the pairs $\tuple{a,b}$ and $\tuple{c,d}$:
\begin{enumerate}
\item  have no common equalisers,
\item or have the same set of equalisers.
\end{enumerate}
\end{lemma}
\begin{proof}
Let $\tuple{u,v}$ be a common equaliser for $\tuple{a,b}$ and $\tuple{c,d}$. In view of Lemma~\ref{lemma:mge_equivalent}
there are elements $s,t\in M$ such that:
\begin{eqnarray*}
u=\frac{a\vee b}{a}s=\frac{c\vee d}{c}t\\
v=\frac{a\vee b}{b}s=\frac{c\vee d}{d}t
\end{eqnarray*} 
Now, it is readily seen that $c \frac{a\vee b}{a}s=d \frac{a\vee b}{b}s$. By the RC-axiom we conclude that
$c \frac{a\vee b}{a} =d \frac{a\vee b}{b}$. Invoking Lemma~\ref{lemma:mge_equivalent} we deduce that there is some $x\in M$ with
$\frac{a\vee b}{a} =\frac{c\vee d}{c}x$ and $\frac{a\vee b}{b} =\frac{c\vee d}{d}x$. Dual argument reveals that there is some $y\in M$ with $\frac{a\vee b}{a}y =\frac{c\vee d}{c}$ and $\frac{a\vee b}{b}y =\frac{c\vee d}{d}$. Comparing both pairs of equalities we see that $xy=yx=e$ and consequently $\tuple{a,b}$ and $\tuple{c,d}$ have the same mge's. Now the result follows by Lemma~\ref{lemma:mge_equivalent}. 
\end{proof}

We conclude this section with an useful observation concerning the mge monoids and infimums of sets.
\begin{lemma}\label{lemma:associativity}
Let ${\cal M}$ be an mge monoid, $\emptyset\subsetneq S\subseteq M$ and $v\in M$ be arbitrary. Then:
\begin{equation*}
\inf (vS) = v\inf S.
\end{equation*}
\end{lemma}
\begin{proof}
First note that if $l\in low(S)$, then $l\le_M s$ for each $s\in S$ and consequently $vl\le_M vs$. This implies that $vl\in low(vS)$.
Therefore $v low(S) \subseteq low(v S)$.

Let $m\in \inf (vS)$. Since $S\neq \emptyset$, there is an $s\in S$ with $m\le_M v s$. This shows that $u=m\vee v$ is defined
and furthermore for every $s\in S$ it holds that $u=m\vee v\le_M v s$. This shows that $u\in low(v S)$ and since $m\in \inf(vS)$ we conclude
that $u\le_M m$. On the other hand, by the definition of $u$, we have that $m\le_M u$ and thus $u\sim_M m$. In particular, $v\le_M m$.
Now since $m\le_M vs $ and $v\le_M m$, we conclude that $\frac{m}{v}\le_M s$ for all $s\in S$. Consequently $\frac{m}{v}\in low(S)$.
Therefore $m=v\frac{m}{v}\in v low(S)$. Since $v low(S)\subseteq low(vS)$ and $m\in \inf(vS)\cap v low(S)$, we deduce that $m\in v \inf S$. 

Conversely, let $m\in v \inf S$. In particular, $m\in v low(S)$ and by above $m\in low(vS)$. Let $u\in low(vS)$ be arbitrary and let $m_0=\frac{m}{v}\in \inf S$. We prove that
$u\le_M m$, which would imply that $m\in \inf(v S)$. Let $u\in low(vS)$. Thus, for any $s\in S$ we have that $u\le_M v s$. Since $S\neq \emptyset$,
such an element $s$ exists and it witnesses that $u\vee v\le_M v s$. Consequently, for each $s\in S$ we have $\frac{u\vee v}{v} \le_M s$ and therefore $\frac{u\vee v}{v} \in low(S)$. This implies that $\frac{u\vee v}{v}\le_M m_0$ and therefore:
\begin{equation*}
u\le_M (u\vee v) = v\frac{u\vee v}{v}\le_M v m_0=v\frac{m}{v}=m.
\end{equation*} 
Therefore $m$ is an upper bound for $low(v S)$ and hence $m\in \inf vS$ as required.
\end{proof}

\subsection{Greatest Common Left Factor Axioms}\label{subsec:gclf}
\begin{definition}\label{LSL-axiom}
Let ${\cal M}$ be a monoid. We say that ${\cal M}$ satisfies the Lower Semi-Lattice axiom (LSL-axiom),
if:
\begin{equation*}
\forall a,b \in {\cal M} \inf (\{a,b\})\neq \emptyset.
\end{equation*}
In this case, we shall denote with $a\maxpref b$ some arbitrary but fixed element of $\inf(\{a,b\})$. For a sequence of elements,
$\{a_i\}_{i=1}^n$ we use $\bigmaxpref_{i=1}^n a_i$ to mean:
\begin{equation*}
\bigmaxpref_{i=1}^n a_i =(\dots ( (a_1\maxpref a_2)\maxpref a_3)\dots ) \maxpref a_n.
\end{equation*}
\end{definition}
\begin{definition}\label{GCLF-axiom}
Let ${\cal M}$ be a monoid. We say that ${\cal M}$ satisfies the Greatest Common Left Factor axiom (GCLF-axiom) if:
\begin{equation*}
\forall m,v,x\in {\cal M} ( m\le_M v \text{ and } m\le_M xv \Rightarrow m\le_M xm).
\end{equation*}
\end{definition}
\begin{remark}
Note that the extreme cases, i.e. $m=v$ and $x=e$, are always satisfied.
\end{remark}
\begin{lemma}\label{GCLF,LSL-Sequential}
If ${\cal M}$ is a sequentiable structure, then ${\cal M}$ satisfies LSL-axiom and GCLF-axiom.
\end{lemma}
\begin{proof}
If ${\cal M}$ is sequentiable, then ${\cal M}$ satisfies LSL-axiom by definition. As for the GCLF-axiom, 
let $m,v,x$ be such that $m\le_M v$ and $m\le_M xv$. From the first inequality we deduce
that $v = mt$ and therefore from the second we have $m\le_M xm t$. Since $\|m\|\le \|xm\|$ by
the definition of a sequentiable structure we get that $m\le_M xm$. 
\end{proof}
\begin{lemma}\label{GCLF,LSL-group}
If ${\cal M}$ is a group, then ${\cal M}$ satisfies LSL-axiom and GCLF-axiom.
\end{lemma}
\begin{proof}
Trivial, since any two elements in a group are in $\sim_{M}$.
\end{proof}
\begin{lemma}\label{LSL-product}
If ${\cal M}_1$ and ${\cal M}_2$ are monoids satisfying the LSL-axiom, so does their Cartesian Product, ${\cal M}_1\times {\cal M}_2$.
\end{lemma}
\begin{proof}
For any two elements $a=\tuple{a_1,a_2}$ and $b=\tuple{b_1,b_2}$ in ${\cal M}_1\times {\cal M}_2$ it holds that:
\begin{equation*}
\inf(\{a,b\}) =\inf(\{a_1,b_1\}) \times \inf(\{a_2,b_2\}).
\end{equation*}
Hence the result. 
\end{proof}
\begin{lemma}\label{GCLF-product}
If ${\cal M}_1$ and ${\cal M}_2$ are monoids satisfying the GLCF-axiom, so does their Cartesian Product, ${\cal M}_1\times {\cal M}_2$.
\end{lemma}
\begin{proof}
Let $m=\tuple{m_1,m_2}$, $v=\tuple{v_1,v_2}$ and $x=\tuple{x_1,x_2}$ be elements in ${\cal M}={\cal M}_1\times {\cal M}_2$.
If $m\le_M v$ and $m\le_M xv$, we get that $m_i\le_{M_i} v_i$ and $m_i\le_{M_i} x_iv_i$ for $i=1,2$. Since ${\cal M}_i$ satisfies GCLF-axiom, it follows that $m_i\le_{M_i} x_im_i$ for $i=1,2$. Therefore $m\le_M xm$. 
\end{proof}
\begin{remark}\label{prevRes}
In a previous work,~\cite{LATA2018}, we showed that for any mge monoid with GCLF- and LSL-axioms the following two results hold true:
\begin{enumerate}
\item for any one-letter transducer ${\cal T}$, there is an equivalent onward transducer with the same states and input-transitions.
\item any subsequential transducer ${\cal T}$ can be minimised.
\end{enumerate}
Furthermore, in~\cite{LATA2018} we provided constructive proofs for these two results. Finally, we showed that the last axiom, GCLF-axiom, is in a way necessary. That is, there is an mge monoid with LSL-axiom that violates the GCLF-axiom and for which a very simple regular language, $a^*b$ does not possess an infimum.
\end{remark}

\subsection{Limit Prefix Axiom}\label{subsec:lpa}
In view of Remark~\ref{prevRes} it is interesting to characterise the (sub)sequential rational functions in terms of congruence relations.
In particular we are interested in a result of the form:

\emph{Given a monoid ${\cal M}$ with certain axioms a function $f:\Sigma^*\rightarrow {\cal M}$ induces a Myhill-Nerode relation, $\equiv_f$, of finite index, $ind(\equiv_f)$, iff there is a complete subsequential transducer with $ind(\equiv_f)$ states that represents $f$.}

The axioms considered in the previous paragraph seem to be not powerful enough to this end. We are not able to prove this formally. Yet, in Section~\ref{Example} we shall give an formal evidence that such a result requires in great extent the properties of the axioms that we consider in this section. 

Particularly, in this section we consider some additional axioms that we refer to as limit prefix axioms. In the next section we shall prove that any of them, actually the weakest of them, suffices to prove the result we stated informally above. Finally, in Section~\ref{Example} we shall prove that a non-uniform version of this axiom must always hold, should the characterisation we are looking at is possible.   

\begin{definition}
For elements $u,v\in M$ we say that $u,v$ have the LP-property if:
\begin{equation*}
\exists \{a_n\}_{n=0}^{\infty}\forall n\in\mathbb{N} (u a_{n+1}=v a_n))\Rightarrow u\le_M v. 
\end{equation*}
In this case write $LP(u,v)$.
\end{definition}
\begin{definition}
For a monoid ${\cal M}$ and elements $u,v\in M$, we define the property $LP_k$ inductively on $k\in \mathbb{N}$:
\begin{enumerate}
\item $LP_0(u,v)=LP(u,v)$.
\item $LP_{k+1}(u,v) \iff u\vee v=\bot \text{ or } LP_{k}\left(\frac{u\vee v}{v},\frac{u\vee v}{u}\right)$.
\end{enumerate}
\end{definition}
\begin{remark}\label{remark:inherit}
Note that $LP_{k}(u,v)\Rightarrow LP_{k+1}(u,v)$ for all $k$. Let us consider the special case where $LP(u,v)=LP_0(u,v)$. We prove that $LP_1(u,v)$ holds. Assume that $u\vee v\neq \bot$, otherwise the statement is obvious. Let $\{a_n\}_{n=0}^{\infty}$ be such that that $\frac{u\vee v}{v}a_{n+1}= \frac{u\vee v}{u}a_{n}$. Let $b_n=\frac{u\vee v}{u}a_{n}$. Thus, we have that:
\begin{equation*}
u b_{n+1} = (u\vee v) a_{n+1}= v\frac{u\vee v}{v}a_{n+1} =v \frac{u\vee v}{u}a_{n}=v b_n.
\end{equation*}
Since $LP(u,v)$ we conclude that $u\le_M v$ and therefore $u\vee v\sim_M v$ which implies that $\frac{u\vee v}{v}\sim_M e\le_M \frac{u\vee v}{u}$.

Now the general case, $LP_k(u,v)\Rightarrow LP_{k+1}(u,v)$ follows by induction on $k$.
\end{remark}
\begin{definition}\label{LP-axiom}
We say that a monoid ${\cal M}$ satisfies the Limit Prefix Axiom (LP-axiom) if:
\begin{equation*}
\forall u,v\in { M} (LP_0(u,v)). 
\end{equation*}
\end{definition}
\begin{definition}\label{ILP-axiom}
We say that a monoid ${\cal M}$ satisfies the Inheritent Limit Prefix Axiom (ILP-axiom) if:
\begin{equation*}
\forall u,v\in { M} \exists k\in \mathbb{N} (LP_k(u,v)). 
\end{equation*}
\end{definition}
\begin{remark}
By Remark~\ref{remark:inherit}, every monoid that satisfies the Prefix Limit Axiom also satisfies the Inherent Limit Prefix Axiom.
\end{remark}
\begin{lemma}\label{LP-Sequential}
If ${\cal M}$ is a sequentiable structure, then ${\cal M}$ satisfies the LP-axiom.
\end{lemma}
\begin{proof}
Let ${\cal M}$ be a sequentiable structure and $ua_{n+1}=va_n$ for some infinite sequence $\{a_n\}_{n=0}^{\infty}$.
Since $u\le_M va_0$ it is enough to prove that $\|u\|\le \|v\|$. However, since $\|ua_{n+1}\|=\|va_n\|$ for
all $n$, we have that $\|a_{n+1}\|-\|a_n\|=\|v\|-\|u\|$. Therefore we have that:
\begin{equation*}
\|a_{n+1}\| - \|a_0\| =(n+1)(\|v\| - \|u\|).
\end{equation*}
Consequently, if $\|v\|<\|u\|$ we get that $\lim_{n\rightarrow\infty} \|a_n\| =-\infty$ whereas $\|a_n\|\ge 0$ by definition.
This proves that $\|u\| \le \|v\|$ and hence $u\le_M v$. 
\end{proof}
\begin{lemma}\label{LP-Groups}
If ${\cal M}$ is a group, then ${\cal M}$ satisfies the LP-axiom.
\end{lemma}
\begin{proof}
Immediate. 
\end{proof}
\begin{definition}\label{WLP-axiom}
For a monoid ${\cal M}$ and elements $u,v,x\in M$ we define the predicate $WLP(u,v,x)$ as:
\begin{equation*}
WLP(u,v,x) \iff \forall  \{a_n\}_{n=0}^{\infty}[\forall n (ua_{n+1}=v{a_n})] \Rightarrow(ux \le_M vx \& \forall n[x\le_M a_n])
\end{equation*}

We say that a monoid ${\cal M}$ satisfies the Weak Limit Prefix Axiom (WLP-axiom) if:
\begin{equation*}
\forall u,v \in M \exists x \in M(WLP(u,v,x)).
\end{equation*}
\end{definition}
\begin{lemma}\label{lemma:mge_inherit_weak}
If an mge monoid ${\cal M}$ satisfies the Inherent Limit Prefix Axiom, then ${\cal M}$ satisfies the Weak Limit Prefix Axiom.
\end{lemma}
\begin{proof}
Assume that ${\cal M}$ satisfies the ILP-Axiom. We prove that for each $k$, $LP_k(u,v)$ implies $\exists x(WLP(u,v,x))$.
This statement is obvious for elements $u,v\in M$ such that there is no sequence $\{a_n\}$ with the property:
\begin{equation*}
u a_{n+1} =v a_n \text{ for all } n\in \mathbb{N}.
\end{equation*}
Indeed, if this is the case $WLP(u,v,x)$ is true for any $x\in M$.
With this remark, we prove that:
\begin{equation*}
LP_k(u,v) \Rightarrow \exists x (WLP(u,v,x)).
\end{equation*}
by induction on $k$. The statement is obvious for $k=0$. Indeed in this case, by above, we may assume that an appropriate sequence
$\{a_n\}$ with $ua_{n+1}=va_n$ exists. Then, by $LP_0(u,v)$ it follows that $u\le_M v$ which means that $WLP(u,v,e)$ is true.
Assume that the above statement holds true for some $k$ and all $u,v\in M$. Let $LP_{k+1}(u,v)$ and $\{a_n\}_{n=0}^{\infty}$ be a sequence with:
\begin{equation*}
ua_{n+1}=va_n.
\end{equation*}
Since ${\cal M}$ is an mge monoid, this shows that $u\vee v$ is well defined. Thus $LP_k(\frac{u\vee v}{v},\frac{u\vee v}{u})$ and by the induction hypothesis we may assume that $WLP(\frac{u\vee v}{v},\frac{u\vee v}{u},y)$ for some $y\in M$. Furthermore since $\tuple{a_{n+1},a_{n}}$ is an equaliser for $\tuple{u,v}$, Lemma~\ref{lemma:mge_equivalent} implies that for each $n$ there is a $b_n$ with:
\begin{equation*}
\frac{u\vee v}{u} b_{n} = a_{n+1} \text{ and } \frac{u\vee v}{v}b_n = a_n.
\end{equation*}
In particular, $\frac{u\vee v}{u} b_{n} = \frac{u\vee v}{v}b_{n+1}$. Now, the existence of the sequence $\{b_n\}_{n=0}^{\infty}$ and $WLP(\frac{u\vee v}{v},\frac{u\vee v}{u},y)$ imply that:
\begin{equation*}
\frac{u\vee v}{v} y\le_M \frac{u\vee v}{u} y.
\end{equation*}
Finally, we multiply the last inequality by $v$ on the left hand side and obtain:
\begin{equation*}
u\frac{u\vee v}{u}y=(u\vee v) y\le v\frac{u\vee v}{u}y.
\end{equation*}
Setting $x=\frac{u\vee v}{u}y$ we get $WLP(u,v,x)$ and the induction step is complete.
\end{proof}

\begin{lemma}\label{WLP-product}
If ${\cal M}_1$ and ${\cal M}_2$ are monoids satisfying the WLP-axiom (ILP-axiom, LP-axiom), so does their Cartesian Product ${\cal M}={\cal M}_1\times {\cal M}_2$.
\end{lemma}
\begin{proof}
We prove that if ${\cal M}_i$ satisfy the WLP-axiom, so does ${\cal M}$.
Let $u=\tuple{u',u''}$ and $v=\tuple{v',v''}$ be elements in ${\cal M}$. Let $x'\in M_1$ and $x''\in M_2$ be such that $WLP(u',v',x')$ and $WLP(u'',v'',x'')$ are satisfied. Let $x=\tuple{x',x''}$. Consider an arbitrary sequence $\{a_n\}_{n=0}^{\infty}$ in ${\cal M}$ with:
\begin{equation*}
ua_{n+1} = va_n \text{ for all } n.
\end{equation*}
Thus, writing $a_n=\tuple{a_n',a_n''}$ we have that $u'a_{n+1}'=v'a_n'$. Thus, by the definition of $x'$ we get that $u'x'\le_{M_1} v'x'$ and $x'\le_{M_1} a_n'$. Similar reasoning shows that $u''x''\le_{M_2} v''x''$ and $x''\le_{M_2} a_n''$. Therefore $ux\le_M vx$ and $x\le_M a_n$ for all $n$. 
Therefore $WLP(u,v,x)$ as required.

The statement for the LP-axiom can be proven analogously. As for the ILP-axiom, we can take into account Remark~\ref{remark:inherit}.
\end{proof}

\section{Characterisation of Sequential Rational Functions}\label{sec:MNR}
This section describes our main contribution, the characterisation of (sub)sequential rational functions.
We start with the definition of the relation, $\equiv_f$, for arbitrary function $f:\Sigma^*\rightarrow {\cal M}$.
Then we state our main result for mge monoids with GCLF- and WLP-axioms in Theorem~\ref{th:Myhill-Nerode}.
The main body of this section is devoted to the proof of this theorem. Lemmata~\ref{lemma:finite_uniformisation} and~\ref{lemma:cancellation}
are the main ingredients to this end.
\begin{definition}\label{MN-Relation}
Let ${\cal M}$ be a monoid and $f:\Sigma^* \rightarrow {\cal M}$ be a function.
For words $\alpha,\beta \in \Sigma^*$ we define $\alpha\equiv_f \beta$ if there exist
$u,v\in {\cal M}$ such that:
\begin{enumerate}
\item $\alpha^{-1} \Dom(f) = \beta^{-1}\Dom(f)$.
\item for all $z\in \alpha^{-1} \Dom(f)$, $\frac{f(\alpha z)}{u}$ and $\frac{f(\beta z)}{v}$ are both defined and
$\frac{f(\alpha z)}{u}=  \frac{f(\beta z)}{v}$.
\end{enumerate}
\end{definition}
\begin{remark}
Note that the second condition for $\alpha\equiv_f v$ can be restated as follows. There exist
$u,v\in {\cal M}$ and a function $s:\Sigma^* \rightarrow {\cal M}$ such that:
\begin{enumerate}
\item[2.$'$] for all $z\in \alpha^{-1} \Dom(f)$ it holds:
\begin{eqnarray*}
f(\alpha z) &= & u s(z) \\
f(\beta z) & = & v s(z).
\end{eqnarray*}
\end{enumerate}
Actually $s(z)=\frac{f(\alpha z)}{u}=  \frac{f(\beta z)}{v}$.

We call a triple $\tuple{u,v,s}$ with the above properties a witness for $\alpha\equiv_f \beta$. For our considerations this perspective is notationally more convenient. For this reason in the sequel we shall use it instead of the more common Definition~\ref{MN-Relation}.
\end{remark} 

In Lemma~\ref{lemma:NM-equivalence}, below we are going to prove that $\equiv_f$ is a an equivalence relation.
With this remark, we can state the main result in this section. It is a characterisation of the subsequential functions over a large class of monoids. It generalises the Myhill-Nerode's Theorem as follows:
\begin{theorem}\label{th:Myhill-Nerode}
Let ${\cal M}$ be an mge-monoid with GCLF-, and WLP-axioms. Let $f:\Sigma^*\rightarrow {\cal M}$ be a function.
Then the following are equivalent:
\begin{enumerate}
\item $\equiv_f$ has finite index.
\item there is a (sub)sequential transducer ${\cal T}$ with ${\cal O}_{\cal T}=f$.
\end{enumerate}
Furthermore if $ind(\equiv_f)=n$, then:
\begin{enumerate}
\item there is a complete (sub)sequential transducer, ${\cal T}$, with $n$ states s.t. ${\cal O}_{\cal T}=f$.
\item any complete (sub)sequential transducer, ${\cal T}$, s.t. ${\cal O}_{\cal T}=f$ has at least $n$ states.
\end{enumerate}
\end{theorem}

We start by proving that for all mge-monoids $\equiv_f$ is an equivalence relation and thus speaking of its index makes perfect sense in Theorem~\ref{th:Myhill-Nerode}. Specifically, we have:
 \begin{lemma}\label{lemma:NM-equivalence}
 Let ${\cal M}$ be an mge-monoid and $f:\Sigma^*\rightarrow {\cal M}$ be a function. 
 Then the relation $\equiv_f$ is a right-invariant equivalence relation.
 \end{lemma}
 \begin{proof}
It is obvious that $\equiv_f$ is reflexive and symmetric. We prove that $\equiv_f$ is also
transitive. Let $\alpha\equiv_f \beta$ and $\beta\equiv_f \gamma$. We show that $\alpha\equiv_f \gamma$.
Since $\alpha^{-1}\Dom(f) = \beta^{-1}\Dom(f)$ and $\beta^{-1}\Dom(f)=\gamma^{-1}\Dom(f)$ we get
$\alpha^{-1}\Dom(f)=\gamma^{-1}\Dom(f)$. Consequently, if $\alpha^{-1}\Dom(f)=\emptyset$, then any triple
$\tuple{u,v,s}$ witnesses that $\alpha\equiv_f \gamma$. 

Thus we assume that $\alpha^{-1}\Dom(f)$ is not empty. Let us fix an element $z_0\in \alpha^{-1}\Dom(f)$.
Consider witnesses $\tuple{u_{\alpha},v_{\alpha},s_{\alpha}}$ for $\alpha\equiv_f \beta$ and $\tuple{u_{\gamma},v_{\gamma},s_{\gamma}}$ for $\gamma\equiv_f{\beta}$. It follows that:
\begin{equation*}
f(\beta z_0)= v_{\alpha} s_{\alpha}(z_0) = v_{\gamma} s_{\gamma}(z_0).
\end{equation*}
Hence $f(\beta z_0)$ is an upper bound for $\{v_{\alpha},v_{\gamma}\}$. By the RMGE-axiom, we have that
$v=v_{\alpha}\vee v_{\gamma}$ is defined. Let $m_{\alpha}=\frac{v}{v_{\alpha}}$ and $m_{\gamma}=\frac{v}{v_{\gamma}}$.
Now, since for all $z\in \beta^{-1} \Dom(f)$ we have:
\begin{equation*}
f(\beta z) =v_{\alpha} s_{\alpha}(z) = v_{\gamma} s_{\gamma}(z)
\end{equation*}
by Lemma~\ref{lemma:mge_equivalent} we have that $\frac{s_{\alpha}(z)}{m_{\alpha}}=\frac{s_{\gamma}(z)}{m_{\gamma}}$.
Let $\hat{s}(z)=\frac{s_{\alpha}(z)}{m_{\alpha}}$ for all $z\in \alpha^{-1}\Dom(f)$. With this remark it is straightforward to see that
$\tuple{u_{\alpha}m_{\alpha},u_{\gamma}m_{\gamma},\hat{s}}$ is a witness for $\alpha\equiv_f \gamma$. Indeed:
\begin{eqnarray*}
f(\alpha z) &=&u_{\alpha} s_{\alpha}(z)=u_{\alpha} m_{\alpha}\frac{s_{\alpha}(z)}{m_{\alpha}}=u_{\alpha} m_{\alpha} \hat{s}(z)\\
f(\gamma z) & = & u_{\gamma}s_{\gamma}(z) = u_{\gamma} m_{\gamma}\frac{s_{\gamma}(z)}{m_{\gamma}}=u_{\gamma} m_{\gamma}\frac{s_{\alpha}(z)}{m_{\alpha}}=u_{\gamma} m_{\gamma} \hat{s}(z)
\end{eqnarray*}
for any $z\in \alpha^{-1}\Dom(f)=\beta^{-1}\Dom(f)=\gamma^{-1}\Dom(f)$.

This proves that $\equiv_f$ is an equivalence relation. Next, we show that it is right invariant.
Let $\alpha\equiv_f \beta$ and $a\in \Sigma$. It is obvious that:
\begin{equation*}
(\alpha a)^{-1} \Dom(f) = a^{-1} \alpha^{-1}\Dom(f) = a^{-1} \beta^{-1}\Dom(f)=(\beta a)^{-1}\Dom(f).
\end{equation*}
Again, if $\alpha^{-1}\Dom(f)=\emptyset$, then $(\alpha a)^{-1}\Dom(f)=\emptyset$ and we are done.
Alternatively, consider a witness $\tuple{u,v,s}$ for $\alpha\equiv_f \beta$. We set $s'(z) =s(az)$ for $z\in \Sigma^*$
and prove that $\tuple{u,v,s'}$ is a witness for $\alpha a\equiv_f \beta a$. Indeed, let $z\in (\alpha a)^{-1}\Dom(f)$.
Thus, $az\in \alpha^{-1}\Dom(f)$ and therefore:
\begin{eqnarray*}
f(\alpha az) &= & u s(az) = u s'(z)\\
f(\beta az) & = & v s(az) = v s'(z)
\end{eqnarray*}
which concludes the proof. 
 \end{proof}
 
In the sequel, we shall delve into the proof of Theorem~\ref{th:Myhill-Nerode}. We start by its easy part. Specifically:
\begin{lemma}\label{lemma:NM-easy}
Let ${\cal M}$ be an mge monoid and ${\cal T}=\tuple{\Sigma,{\cal M},Q,i,F,\delta,\lambda,\iota,\Psi}$ be a complete (sub)sequential transducer. If $f={\cal O}_{\cal T}$ then $ind(\equiv_f)\le |Q|$.
\end{lemma}
\begin{proof}
Let us define $\sim_{\cal T}\subseteq \Sigma^*\times \Sigma^*$ as:
\begin{equation*}
\alpha\sim_{\cal T} \beta \iff \delta^*(i,\alpha)=\delta^*(i,\beta).
\end{equation*}
Since $\delta$ is a total function, $\sim_{\cal T}$ is reflexive. The symmetry and transitivity are apparent. Therefore $\sim_{\cal T}$ is an equivalence relation. We prove that $\sim_{\cal T}\subseteq \equiv_f$. This would imply that $ind(\equiv_f)\le ind(\sim_{\cal T})$. Since, obviously, $|\sim_{\cal T}|\le |Q|$ the result would follow.

To complete the proof, we show that if $\alpha \sim_{\cal T}\beta$ then $\alpha\equiv_f \beta$. Let $p\in Q$ be such that:
\begin{equation*}
p=\delta^*(i,\alpha)=\delta^*(i,\beta).
\end{equation*}
Since $\delta$ is a function, it is clear that $\alpha^{-1}\Dom(f) =\beta^{-1}\Dom(f)$ and more specifically we have:
\begin{equation*}
\alpha^{-1}\Dom(f) =\beta^{-1}\Dom(f) = \{\gamma\in \Sigma^*\,|\, \delta^*(p,\gamma)\in F\}
\end{equation*}
Next we introduce:
\begin{eqnarray*}
u & = & \iota \circ \lambda^*(i,\alpha)\\
v & = & \iota \circ \lambda^*(i,\beta) \\
s(\gamma) & = & \lambda^*(p,\gamma) \Psi(\delta^*(p,\gamma)).
\end{eqnarray*}
We claim that $\tuple{u,v,s}$ is a witness for $\alpha\equiv_f \beta$. Indeed let $\gamma\in \alpha^{-1}\Dom(f) $. From above we have that this is equivalent to $\delta^*(p,\gamma)\in F$. Furthermore we have:
\begin{eqnarray*}
f(\alpha\gamma) &=&{\cal O}_{\cal T}(\alpha\gamma)=\iota \circ \lambda^*(i,\alpha\gamma)\Psi(\delta^*(i,\alpha\gamma))\\
f(\beta\gamma) &=&{\cal O}_{\cal T}(\beta\gamma)=\iota \circ \lambda^*(i,\beta\gamma)\Psi(\delta^*(i,\beta\gamma))
\end{eqnarray*}
However, it is obvious that:
\begin{eqnarray*}
\iota \circ \lambda^*(i,\alpha\gamma)\Psi(\delta^*(i,\alpha\gamma)) &= & \iota \circ \lambda^*(i,\alpha)\circ \lambda^*(p,\gamma)\circ \Psi(\delta^*(p,\gamma)) = u s(\gamma) \\
\iota \circ \lambda^*(i,\beta\gamma)\Psi(\delta^*(i,\beta\gamma)) &= & \iota \circ \lambda^*(i,\beta)\circ \lambda^*(p,\gamma)\circ \Psi(\delta^*(p,\gamma)) = v s(\gamma),
\end{eqnarray*}
which concludes the proof. 
\end{proof}

The rest of this section is devoted to the non-trivial part of Theorem~\ref{th:Myhill-Nerode}. Specifically, we want to show that whenever $\equiv_f$ has a finite index we can construct a (sub)sequential transducer with $ind(\equiv_f)$ states recognising $f$.
The problem here arises from the fact that we have no explicit information about the output language, $\Codom(f)$. Indeed, the functions $s$ that 
determine the witnesses can be arbitrary and it is by far not obvious that even their range should be regular over ${\cal M}$. It is due to the axioms 
GCLF and WLP that we are going to extract some information about the witnesses and use it to define the desired (sub)sequential transducer.
It is interesting to note that in the absence of the LSL-axiom, we also do not have infimums for every pair of monoid elements. Thus, we cannot claim that for every regular language over ${\cal M}$ possesses an infimum. Consequently, the classical idea that the witnesses $\tuple{u,v,s}$ for $\alpha\equiv_f \beta$ should/can be selected as:
\begin{eqnarray*}
u&=&\inf\{ f(\alpha z) \,|\, z\in \alpha^{-1}\Dom(f)\}\\
v&=&\inf\{ f(\beta z)   \,|\, z\in \beta^{-1}\Dom(f)\}
\end{eqnarray*}
fails in the very beginning. Leave alone the fact that these two sets should not be regular.

We start our study of the problem by showing the following important implication of the WLP-axiom.
\begin{lemma}\label{lemma:finite_uniformisation}
Let $f:\Sigma^*\rightarrow {\cal M}$ be a function in an mge-monoid satisfying the GCLF- and WLP-axioms. Let $\alpha_1,\dots,\alpha_N$ be pairwise equivalent
with respect to $\equiv_f$. Then, there are elements $v_1,\dots,v_N\in {\cal M}$ and $s:\Sigma^*\rightarrow {\cal M}$ such that:
\begin{enumerate}
 \item $\tuple{v_i,v_j,s}$ is a witness for $\alpha_i\equiv_f \alpha_j$.
 \item if $\alpha_i\le_{\Sigma^*}\alpha_j$ then $v_i\le_M v_j$.
 \end{enumerate}
\end{lemma}
\begin{proof}
The claim is trivial if $\alpha_1^{-1}\Dom(f)=\emptyset$. Alternatively, let us fix an element $\gamma_0\in \alpha_1^{-1}\Dom(f)$. Since $\alpha_1\equiv_f \alpha_i$ for $i\le N$ there is a witness $\tuple{u_i,w_i,s_i}$ for $\alpha_1\equiv_f \alpha_i$. Now we have that:
\begin{equation*}
f(\alpha_1 \gamma_0) = u_i s_i(\gamma_0)
\end{equation*}
for each $i\le N$. Therefore $u_is_i(\gamma_0) \in up(\{u_1,\dots,u_n\})$. Let $u=\bigvee_{i=1}^n u_i$. It follows that for each $\gamma\in \alpha_1^{-1}\Dom(f)$, $u_i\le_M u\le_M u_is_i(\gamma)$. Hence we can define $\hat{s}:\alpha_1^{-1}\Dom(f)\rightarrow M$ as:
\begin{equation*}
\hat{s}(\gamma) = \frac{u_1s_1(\gamma)}{u}=\frac{u_is_i(\gamma)}{u}=\frac{s_i(\gamma)}{\frac{u}{u_i}}.
\end{equation*}
We set $v_i'=w_i \circ \frac{u}{u_i}$. Now it is clear that:
\begin{equation*}
f(\alpha_i\gamma) = w_i s_i(\gamma) = w_i \frac{u}{u_i}\frac{s_i(\gamma)}{\frac{u}{u_i}}=v'_i \hat{s}(\gamma).
\end{equation*}

So far we have that $\tuple{v'_i,v_j',\hat{s}}$ satisfy the first property. We use the WLP-axiom in order to modify these witnesses so that they satisfy the second property as well. To this end, let:
\begin{equation*}
P=\{\tuple{i,j}\,|\, \alpha_i\le_{\Sigma^*}\alpha_j\}.
\end{equation*}
In words, $P$ is the set of all pairs $\tuple{i,j}$ such that $\alpha_i$ is a prefix of $\alpha_j$.
Let us consider an element $\tuple{i,j}\in P$.
Since $\alpha_i$ is a prefix of $\alpha_j$ there is some $\beta$ with $\alpha_i\beta=\alpha_j$. By the right invariance of $\equiv_f$ and since $\alpha_i\equiv_f \alpha_j$, we get that $\alpha_i\beta^k\equiv_f \alpha_i$ for each natural number $k$. In particular, for every $\gamma\in \alpha_i^{-1}\Dom(f)$ we have $\alpha_i\beta^k\gamma\in \Dom(f)$ for each $k\in \mathbb{N}$ and therefore $a_k=\hat{s}(\beta^k\gamma)$ is well-defined.
Now it is easy to see that:
\begin{equation*}
v'_i \hat{s}(\beta^{k+1}\gamma) = f(\alpha_i\beta^{k+1}\gamma)=v'_j \hat{s}(\beta^k\gamma).
\end{equation*}
This shows that for every $k$ it holds $v'_i a_{k+1}=v'_j a_k$. Thus, by the WLP-axiom we conclude that there is some $x_{i,j}\in M$ such that:
\begin{equation*}
v'_i x_{i,j}\le_M v'_j x_{i,j}\text{ and } x_{i,j} \le_M \hat{s}(\gamma) \text{ for } \gamma\in \alpha_i^{-1}\Dom(f).
\end{equation*}
Consequently for every $\gamma\in \alpha_1^{-1}\Dom(f)$ it holds:
\begin{equation*}
\hat{s}(\gamma)\in up\{x_{i,j}\,|\, \tuple{i,j}\in P\}
 \end{equation*}
By the RMGE-axiom and Lemma~\ref{lemma:mge}, $X=\bigvee_{\tuple{i,j}\in P} x_{i,j}$ is well-defined and $X\le_M \hat{s}(\gamma)$ for all $\gamma\in \alpha_i^{-1}\Dom(f)$. Therefore we can define $s:\alpha_1^{-1}\Dom(f)\rightarrow M$ as:
 \begin{equation*}
 s(\gamma) = \frac{\hat{s}(\gamma)}{X} \text{ for } \gamma\in \alpha_i^{-1}\Dom(f).
 \end{equation*}
Finally, we set $v_i=v_i'X$ for $i\le N$. A straightforward computation shows that $\tuple{v_i,v_j,s}$ is a witness for $\alpha_i\equiv_f \alpha_j$.
 
 It remains to be shown that $\alpha_i\le_{\Sigma^*}\alpha_j$ always implies $v_i\le_M v_j$. Let $\alpha_j=\alpha_i\beta$. Hence $\tuple{i,j}\in P$ and  we have that $x_{i,j}\le_M X$ and furthermore, by $v'_ix_{i,j}\le_M v_j'x_{i,j}$, $v'_i x_{i,j} t =v'_j x_{i,j}$ for some $t$.
 Next, considering $f(\alpha_i\beta^{k+1}\gamma_0)$, we get:
 \begin{eqnarray*}
 v'_i x_{i,j} \frac{X}{x_{i,j}} s(\beta^{k+1}\gamma_0) &=&f(\alpha_i\beta^{k+1}\gamma_0)\\
 &=& v'_j x_{i,j} \frac{X}{x_{i,j}} s(\beta^k\gamma_0)\\
 &=&v_i'x_{i,j}t\frac{X}{x_{i,j}} s(\beta^k\gamma_0).
 \end{eqnarray*}
By the LC-axiom, we conclude that $t \frac{X}{x_{i,j}}s(\beta^k\gamma_0) =  \frac{X}{x_{i,j}} s(\beta^{k+1}\gamma_0)$. 
Setting $b_k=s(\beta^k\gamma_0) $ this is equivalent to: $$\frac{X}{x_{i,j}} b_{k+1}= t \frac{X}{x_{i,j}}b_k \text{ for } k\in \mathbb{N}$$ and by the WLP-axiom,
 we conclude that there is some $y$ with $\frac{X}{x_{i,j}} y \le_M t\frac{X}{x_{i,j}} y$.
 Therefore $\frac{X}{x_{i,j}}\le_M t\frac{X}{x_{i,j}} y$ and by the GCLF-axiom we deduce that $\frac{X}{x_{i,j}}\le_M t\frac{X}{x_{i,j}}$.
 Therefore $$v_i = v_i' x_{i,j} \frac{X}{x_{i,j}}\le v_i'x_{i,j}t \frac{X}{x_{i,j}}= v_j'x_{i,j}\frac{X}{x_{i,j}}=v_j$$ as required.
 \end{proof}

\begin{lemma}
Let ${\cal M}$ be an mge-monoid satisfying the GCLF and WLP-axioms. Let $f:\Sigma^*\rightarrow {\cal M}$ be a function with $ind(\equiv_f)=n$. Then there is a complete (sub)sequential transducer with $n$ states that represents $f$.
\end{lemma}
\begin{proof}
Let $C_1,C_2,\dots,C_n$ be the equivalence classes of $\equiv_f$. Let $\alpha_i\in C_i$
be a shortest element of the $i$-th class. Since $\equiv_f$ is right invariant, $|\alpha_i|\le n-1$ 
for each $i$. We set:
\begin{equation*}
A_i = \{\alpha\in C_i \,|\, |\alpha|\le 2n-1\}.
\end{equation*}
Since $2n-1\ge n$ for $n\ge 1$, we have that $\alpha_i\in A_i$. Furthermore, $A_i$ is finite, for $\Sigma$ is finite,
and by Lemma~\ref{lemma:finite_uniformisation} there are elements $v(\beta)\in {\cal M}$ for each $\beta\in A_i$ and
a function $s_i:\Sigma^*\rightarrow {\cal M}$ such that:
\begin{enumerate}
\item $\tuple{v(\alpha_i),v(\beta),s_i}$ is a witness for $\alpha_i\equiv_f \beta$.
\item if $\alpha\beta,\alpha\in A_i$ for some $\alpha,\beta$ then $v(\alpha)\le_M v(\alpha\beta)$.
\end{enumerate}
For $i,j\le n$ we let:
\begin{equation*}
B_{i,j} =\{\beta \,|\, \alpha_i \beta\in A_j\}.
\end{equation*}
Note that $B_{i,j}\subseteq A_j$ and since $A_j$ is finite, it follows that $B_{i,j}$ is also finite.
Let $j\le n$ be such that $\alpha_j^{-1}\Dom(f)\neq \emptyset$. Then, for any $i\le n$ and $\beta\in B_{i,j}$ we have:
\begin{equation*}
f(\alpha_i\beta \gamma) = v(\alpha_i) s_i(\beta\gamma) = v(\alpha_i\beta) s_j(\gamma).
\end{equation*} 
Thus, $f(\alpha_i\beta\gamma)$ is an upper bound for $v(\alpha_i)$ and $v(\alpha_i\beta)$. Hence:
\begin{equation*}
v(\alpha_i) \vee v(\alpha_i\beta) \text{ and } \frac{v(\alpha_i)\vee v(\alpha_i\beta)}{v(\alpha_i\beta)}
\end{equation*}
are well-defined. Furthermore, by the LC-axiom, we get that:
\begin{equation*}
\frac{v(\alpha_i)\vee v(\alpha_i\beta)}{v(\alpha_i\beta)} \le_M s_j(\gamma).
\end{equation*}
This allows us to consider the set:
\begin{equation*}
E_j = \left\{\frac{v(\alpha_i)\vee v(\alpha_i\beta)}{v(\alpha_i\beta)}\,|\, i\le n, \beta\in B_{i,j}\right\}.
\end{equation*}
By the above discussion we have that each of the elements in $E_j$ is well-defined and less than or equal to $s_j(\gamma)$.
Consequently $s_j(\gamma) \in up(E_j)$. Since $E_j$ is finite, as $B_{i,j}$ are finite and $i\le n$, by the RMGE-axiom we get that $\sup E_j\neq \emptyset$.
We fix $M_j \in \sup E_j$ for each $j$ such that $\alpha_j^{-1}\Dom(f)\neq\emptyset$.

To conclude the proof we will need the following:
\begin{lemma}\label{lemma:cancellation}
Let ${\cal M}$ be an mge-monoid with GCLF- and WLP-axioms. Let $i,j\le n$ and $a\in \Sigma$ be such that $\alpha_i a\equiv_f\alpha_j$ and $\alpha_j^{-1}\Dom(f)\neq \emptyset$. Then:
\begin{equation*}
v(\alpha_i) M_i \le v(\alpha_i a) M_j.
\end{equation*}
\end{lemma}
Assume that Lemma~\ref{lemma:cancellation} holds. Without loss of generality we assume that $C_1=[\varepsilon]$ and construct a (sub)sequential transducer:
\begin{eqnarray*}
{\cal T} &= & \tuple{\Sigma, {\cal M},{\cal C},C_1,F,v(\varepsilon)M_1,\delta,\lambda,\psi}\\
{\cal C} & = & \{C_i \,|\, i\le n\}\\
F & = & \{C_i \,|\, \alpha_i\in \Dom(f) \} \\
\delta(C_i,a) & = & C_j \iff \alpha_i a\in C_j \\
\lambda(C_i,a) & = &\begin{cases} 
e \text{ if } (\alpha_ia)^{-1}\Dom(f)= \emptyset \\
\frac{ v(\alpha_i a) M_j}{v(\alpha_i) M_i}, \text{ where }\delta(C_i,a)=C_j,\text{ else}
\end{cases}\\
\psi(C_i) & = & \frac{s_i(\varepsilon)}{M_i}.
\end{eqnarray*}
Note that $M_i\le_M s_i(\gamma)$ for any $\gamma\in \alpha_i^{-1}\Dom(f)$. Since $\varepsilon\in \alpha_i^{-1}\Dom(f)$ for each final state $C_i$ we get that $M_i\le_M s_i(\gamma)$ for a final state $C_i$. Hence the function $\Psi$ is well-defined.
By the same argument we can put $s_i(\gamma) = M_i \hat{s}_i(\gamma)$ for $\gamma\in \alpha_i^{-1}\Dom(f)$.
Let $C_i,C_j$ and $a\in \Sigma$ be such that $\delta(C_i,a)=C_j$. Let $\gamma\in \alpha_j^{-1}\Dom(f)$ be arbitrary. Then:
\begin{eqnarray*}
f(\alpha_i a \gamma) &=& v(\alpha_i) s_i( a\gamma) = v(\alpha_i) M_i \hat{s}_i(a\gamma) \\
f(\alpha_i a\gamma) &=& v(\alpha_i a) s_j(\gamma) = v(\alpha_i a) M_j \hat{s}_j(\gamma).
\end{eqnarray*}
Applying Lemma~\ref{lemma:cancellation} and the LC-axiom we get that $\hat{s}_i(a\gamma)=\frac{v(\alpha_i a)M_j}{v(\alpha_i) M_i} \hat{s}_j(\gamma)$. Now a straightforward induction shows that for any $\alpha\gamma\in \Dom(f)$ it holds:
\begin{equation*}
f(\alpha\gamma) = v(\varepsilon)M_1 \lambda^*(C_1,\alpha) \circ \hat{s}_j(\gamma)
\end{equation*}
where $C_j=\delta^*(C_1,\alpha)$. In particular, if $\gamma=\varepsilon$ we get:
\begin{equation*}
f(\alpha)=f(\alpha\gamma) = v(\varepsilon)M_1 \lambda^*(C_1,\alpha) \circ \hat{s}_j(\varepsilon) = v(\varepsilon)M_1 \lambda^*(C_1,\alpha)\Psi(C_j) = f_{\cal T}(\alpha).
\end{equation*}
The fact that the domains of $f$ and ${\cal T}$ coincides is a routine. 
\end{proof}

To complete the proof of Theorem~\ref{th:Myhill-Nerode} we need to establish the truthfulness of Lemma~\ref{lemma:cancellation}.
First we state the following useful observation:
\begin{lemma}\label{lemma:cycle_free}
Let ${\cal M}$ be an mge-monoid with GCLF- and WLP-axioms. Let $\beta=\beta_1\beta_2\beta_3$ be of length $|\beta|\le n$. Let $\alpha_k\beta_1\equiv_f \alpha_i$,
$\alpha_i\beta_2\equiv_f \alpha_i$, and $\alpha_i\beta_3\equiv_f \alpha_j$ with $\alpha_j^{-1}\Dom(f)\neq \emptyset$. 
Then:
\begin{enumerate}
\item $b_{i,j}=\frac{v(\alpha_i)\vee v(\alpha_i\beta_3)}{v(\alpha_i\beta_3)}$ and $b_{k,j}=\frac{v(\alpha_k)\vee v(\alpha_k\beta_1\beta_3)}{v(\alpha_k\beta_1\beta_3)}$ are well defined.
\item $up(\{b_{i,j},b_{k,j}\})\neq\emptyset$.
\item for any $b\in up(\{b_{i,j},b_{k,j}\})$ it holds $v(\alpha_k)\le_M v(\alpha_k\beta) b$.
\end{enumerate}
\end{lemma}
\begin{proof}
Since $\alpha_j^{-1}\Dom(f)\neq \emptyset$ and $|\beta|\le n$, $|\alpha_k|<n$, and $|\alpha_i|<n$ we have that $\beta_3\in B_{i,j}$
and $\beta_1\beta_3\in B_{k,j}$. Therefore, by the definition of $E_j$ we have that $b_{i,j}$ and $b_{k,j}$ are defined and belong to $E_j$. Since $M_j\in \sup E_j$ is an upper bound for all the elements in $E_j$, we conclude that it is also an upper bound for $\{b_{i,j},b_{k,j}\}$. By RMGE-axiom, we have that $\sup(\{b_{i,j},b_{k,j}\})\neq \emptyset$.

To prove the third part of the lemma, we fix an element $\gamma\in \alpha_j^{-1}\Dom(f)$. By the above discussion, we have that 
$\alpha_i\beta_3\in A_j$, $\alpha_k\beta_1\in A_i$, and $\alpha_k\beta_1\beta_3\in A_j$. Putting these together, we get:
\begin{eqnarray*}
f(\alpha_i\beta_3\gamma)  & =  v(\alpha_i) s_i(\beta_3\gamma) \quad f(\alpha_k\beta_1\beta_3\gamma)  = & v(\alpha_k\beta_1) s_i(\beta_3\gamma)\\
f(\alpha_i\beta_3\gamma) & = v(\alpha_i\beta_3) s_j(\gamma) \quad f(\alpha_k\beta_1\beta_3\gamma)  = & v(\alpha_k\beta_1\beta_3) s_j(\gamma).
\end{eqnarray*}
This shows that $\tuple{v(\alpha_i),v(\alpha_i\beta_3)}$ and $\tuple{v(\alpha_k\beta_1),v(\alpha_k\beta_1\beta_3)}$ have a common equaliser, $\tuple{s_i(\beta_3\gamma),s_j(\gamma)}$. Therefore, by Lemma~\ref{lemma:mge_partition}, $\tuple{v(\alpha_i),v(\alpha_i\beta_3)}$ and $\tuple{v(\alpha_k\beta_1),v(\alpha_k\beta_1\beta_3)}$ have the same set of equalisers. This implies that the mge of $\tuple{v(\alpha_i),v(\alpha_i\beta_3)}$ is also an mge for  $\tuple{v(\alpha_k\beta_1),v(\alpha_k\beta_1\beta_3)}$. Consequently,
\begin{equation*}
 b_{i,j} = \frac{v(\alpha_i)\vee v(\alpha_i\beta_3)}{v(\alpha_i\beta_3)} \sim_M \frac{v(\alpha_k\beta_1)\vee v(\alpha_k\beta_1\beta_3)}{v(\alpha_k\beta_1\beta_3)}.
\end{equation*}
This shows that:
\begin{equation*}
v(\alpha_k\beta_1\beta_3) \circ b_{i,j} \sim_M v(\alpha_k\beta_1)\vee v(\alpha_k\beta_1\beta_3).
\end{equation*}

On the other hand, by the definition of $b_{k,j}$, we have:
\begin{equation*}
v(\alpha_k\beta_1\beta_3)\circ b_{k,j} = v(\alpha_k)\vee v(\alpha_k\beta_1\beta_3).
\end{equation*} 
Again, considering $\alpha_k\in A_k$, $\alpha_k\beta_1\in A_i$, and $\alpha_k\beta_1\beta_3\in A_j$ we have; 
\begin{eqnarray*}
f(\alpha_k\beta_1\beta_3\gamma) & = v(\alpha_k) s_k(\beta_1\beta_3\gamma)  \\
f(\alpha_k\beta_1\beta_3\gamma) & =  v(\alpha_k\beta_1) s_i(\beta_3\gamma) \\
f(\alpha_k\beta_1\beta_3\gamma) & = v(\alpha_k\beta_1\beta_3) s_j(\gamma).
\end{eqnarray*}
This shows that $\{v(\alpha_k),v(\alpha_k\beta_1),v(\alpha_k\beta_1\beta_3)\}$ has an upper bound and therefore:
\begin{eqnarray*}
b_j &=& \frac{v(\alpha_k)\vee v(\alpha_k\beta_1)\vee v(\alpha_k\beta_1\beta_3)}{v(\alpha_k\beta_1\beta_3)} \\
b_i &=& \frac{v(\alpha_k)\vee v(\alpha_k\beta_1)\vee v(\alpha_k\beta_1\beta_3)}{v(\alpha_k\beta_1)}
\end{eqnarray*}
are well-defined.
Now, since $v(\alpha_k\beta_1\beta_3) b_j=v(\alpha_k)\vee v(\alpha_k\beta_1)\vee v(\alpha_k\beta_1\beta_3)\sim_{M}v(\alpha_k\beta_1\beta_3)(b_{i,j}\vee b_{k,j})$, it is clear that $b_j\in \sup\{b_{i,j},b_{k,j}\}$. Furthermore, since $v(\alpha_k\beta_1\beta_3)b_j= v(\alpha_k\beta_1)b_i$, by Lemma~\ref{lemma:mge_equivalent} we get that:
\begin{equation*}
\frac{s_i(\beta_3\gamma)}{b_i} = \frac{s_j(\gamma)}{b_j}.
\end{equation*}
Let us denote with $b_i',b_k'$:
\begin{eqnarray*}
b_k' & = & \frac{v(\alpha_k)\vee v(\alpha_k\beta_1)}{v(\alpha_k)}\\
b_i' & = & \frac{v(\alpha_k)\vee v(\alpha_k\beta_1)}{v(\alpha_k\beta_1)}.
\end{eqnarray*}
In particular, $b_i'\le_M b_i$.
Finally, since $\alpha_k\beta_1\equiv_f \alpha_k\beta_1\beta_2$ and both $\alpha_k\beta_1,\alpha_k\beta_1\beta_2\in A_i$, we have that $v(\alpha_k\beta_1\beta_2)=v(\alpha_k\beta_1) x$ for some $x\in M$ and therefore $s_i(\beta_2\beta_3\gamma)=xs_i(\beta_3\gamma)$. Hence:
\begin{eqnarray*}
f(\alpha_k\beta\gamma) &= & v(\alpha_k) s_k(\beta\gamma) \\
f(\alpha_k\beta\gamma) &= & v(\alpha_k\beta_1) s_i(\beta_2\beta_3\gamma) = v(\alpha_k\beta_1) x s_i(\beta_3\gamma) \\
 f(\alpha_k\beta\gamma) &= &v(\alpha_k\beta) s_j(\gamma).
\end{eqnarray*}
By the first and second equalities we get that $b_i'\le_M xs_i(\beta_3\gamma)$. Since we also have $b_i'\le_M s_i(\beta_3\gamma)$, by the GCLF-axiom we deduce that $b_i'\le_M x b_i'$. Finally, by $b_i'\le_M b_i$ it follows that $b_i'\le_M x b_i$. Now we conclude the proof by:
\begin{equation*}
v(\alpha_k) \le_M v(\alpha_k\beta_1) b_i'\le_M v(\alpha_k\beta_1) x b_i.
\end{equation*}
However, we have that:
\begin{equation*}
v(\alpha_k\beta_1) x b_i \frac{s_i(\beta_3\gamma)}{b_i}=v(\alpha_k\beta_1) x s_i(\beta_3\gamma)=v(\alpha_k\beta) s_j(\gamma)=v(\alpha_k\beta)b_j \frac{s_j(\gamma)}{b_j}.
\end{equation*}
Since $\frac{s_i(\beta_3\gamma)}{b_i}=\frac{s_j(\gamma)}{b_j}$ we deduce that:
\begin{equation*}
v(\alpha_k\beta_1) x b_i =v(\alpha_k\beta) b_j
\end{equation*}
and therefore $v(\alpha_k)\le_M v(\alpha_k\beta) b_j$ which proves that $\frac{v(\alpha_k)\vee v(\alpha_k\beta)}{v(\alpha_k\beta)}\le_M b_j$. Since $b_j\in \sup(\{b_{i,j},b_{k,j}\})$, the result follows. 
\end{proof}

\begin{corollary}
Let $B_{i,j}'$ and $E_j'$ be defined as:
\begin{eqnarray*}
B_{i,j}' & = & \{\beta \,|\, |\beta|<n \text{ and } \alpha_i\beta\in A_j\} \\
E_j' & = & \left\{\frac{v(\alpha_i)\vee v(\alpha_i\beta)}{v(\alpha_i\beta)} \,|\, \beta\in B_{i,j}'\right\}.
\end{eqnarray*}
If $\alpha^{-1}_j\Dom(f)\neq\emptyset$, then $M_j\in \sup E_j'$.
\end{corollary}
\begin{proof}
It is clear that $B_{i,j}'\subseteq B_{i,j}$, therefore $E_j'\subseteq E_j$ and, in particular, is finite. 
Therefore $M_j$ is an upper bound for it.
Since ${\cal M}$ is an mge-monoid, the set $E_j'$ has a least upper bound, say $M_j'$. Hence $M_j'\le_M M_j$.
To establish that $M_j\le_M M_j'$ it suffices to prove that $M_j'\in up(E_j)$. We prove that for any $\beta\in B_{i,j}\setminus B_{i,j}'$:
\begin{equation*}
\frac{v(\alpha_i)\vee v(\alpha_i\beta)}{v(\alpha_i\beta)}\le_M M_j'.
\end{equation*} 
For the sake of contradiction, assume that this is not the case and let $\beta$ be of least length such that there exist $k,j\le n$ with the properties:
\begin{enumerate}
\item $\beta\in B_{k,j}$ and
\item 
$\frac{v(\alpha_k)\vee v(\alpha_k\beta)}{v(\alpha_k\beta)}\not \le_M M_j'$.
\end{enumerate}
In particular $\beta\not \in B_{k,j}'$ and thus $|\beta|\ge n$. Hence by the right invariance of $\equiv_f$ we can decompose $\beta=\beta_1\beta_2\beta_3$ such that $\alpha_k\beta_1\equiv_f \alpha_k\beta_1\beta_2$ and $|\beta_2|\ge 1$.  Let $\alpha_i\equiv_f \alpha_k\beta_1$. Now, $|\beta_1\beta_3|< |\beta|$ and $|\beta_3|<|\beta|$. Hence, $\beta_3\in B_{i,j}$ and $\beta_1\beta_3\in B_{k,j}$. By the minimality of $\beta$ we further get that:
\begin{eqnarray*}
b_{k,j}&=&\frac{v(\alpha_k)\vee v(\alpha_k\beta_1\beta_3)}{v(\alpha_k\beta_1\beta_3)} \le_M M_j'\\
b_{i,j} &=& \frac{v(\alpha_i)\vee v(\alpha_i\beta_3)}{v(\alpha_i\beta_3)}\le_M M_j'.
\end{eqnarray*}
Hence $M_j'$ is an upper bound for $b_{k,j}$ and $b_{i,j}$ and by Lemma~\ref{lemma:cycle_free} we conclude that:
\begin{equation*}
\frac{v(\alpha_k)\vee v(\alpha_k\beta)}{v(\alpha_k\beta)}\le_M M_j'
\end{equation*}
contrary to our assumption. Therefore $M_j'$ is an upper bound for $E_j$ and hence $M_j\le_M M'_j$. Summing up we get that
$M_j\sim_M M_j'$ and since $M_j'\in \sup E'_j$ it follows that $M_j\in \sup E_j'$. 
\end{proof}
Now we are ready to prove Lemma~\ref{lemma:cancellation}:
\begin{proof}[of Lemma~\ref{lemma:cancellation}]
First we show that $v(\alpha_i a)M_j\in up(v(\alpha_i)\circ E_i')$. Let $\gamma\in \alpha_j^{-1}\Dom(f)$ be fixed. Let $b_i\in E_i'$. Hence there is $\beta\in B_{k,i}'$ such that:
\begin{equation*}
b_i = \frac{v(\alpha_k)\vee v(\alpha_k\beta)}{v(\alpha_k\beta)}.
\end{equation*}
Since $|\beta|<n$ and $\alpha_i a\equiv_f \alpha_j$, it follows that $|\beta a|\le n$ and therefore $\beta a\in B_{k,j}$. Hence:
\begin{equation*}
b_j =\frac{v(\alpha_k)\vee v(\alpha_k\beta a)}{v(\alpha_k\beta a)} \in E_j.
\end{equation*}
Furthermore, we have that $\alpha_i a\equiv_f \alpha_j$ and $|a|=1\le n$, hence:
\begin{equation*}
c = \frac{v(\alpha_i) \vee v(\alpha_i a)}{v(\alpha_i a)}\in E_j.
\end{equation*}
Now, we have that:
\begin{eqnarray*}
f(\alpha_i a \gamma) & = v(\alpha_i ) s_i(a\gamma) \quad f(\alpha_k\beta a \gamma) & = v(\alpha_k\beta ) s_i(a\gamma)\\
f(\alpha_i a \gamma) & = v(\alpha_i a) s_j(a\gamma) \quad f(\alpha_k\beta a \gamma) & = v(\alpha_k\beta a) s_j(\gamma).
\end{eqnarray*}
This shows that $\tuple{v(\alpha_i),v(\alpha_i a)}$ and $\tuple{v(\alpha_k\beta),v(\alpha_k\beta a)}$ have a common equaliser, $\tuple{s_i(a\gamma),s_j(\gamma)}$. Consequently, by Lemma~\ref{lemma:mge_partition}, we get that the set of their equalisers are the same and therefore:
\begin{equation*}
c \sim_M \frac{v(\alpha_k\beta) \vee v(\alpha_k\beta a)}{v(\alpha_k\beta a)}.
\end{equation*}
Now, by the definition of $M_j$ we have that $c\le_M M_j$ and $b_j\le_M M_j$. Therefore:
\begin{equation*}
v(\alpha_k\beta a) \vee v(\alpha_k\beta) \vee v(\alpha_k)\sim_M v(\alpha_k\beta a)(c\vee b_j) \le_M v(\alpha_k\beta a) M_j. 
\end{equation*}
We conclude that $v(\alpha_k\beta) b_i \le_M v(\alpha_k\beta a) M_j$. Again, since the set of equalisers of $\tuple{v(\alpha_k\beta),v(\alpha_k\beta a)}$ and $\tuple{v(\alpha_i),v(\alpha_i a)}$ coincide, we deduce that:
\begin{equation*}
v(\alpha_i) b_i\sim_M v(\alpha_i) \vee v(\alpha_i a) \le_M v(\alpha_i a) M_j,
\end{equation*}
where the last inequality follows by the fact that $\frac{v(\alpha_i) \vee v(\alpha_i a)}{v(\alpha_i a)}\in E_j$ and $M_j\in \sup E_j$.
Hence by the LC-axiom, $b_i\le_M \frac{v(\alpha_i a) M_j}{v(\alpha_i)}$ for any $b_i\in E_i'$. Since $M_i\in \sup E_i'$, this implies
$M_i \le_M \frac{v(\alpha_i a) M_j}{v(\alpha_i)}$ and multiplying by $v(\alpha_i)$ on left hand side we obtain:
\begin{equation*}
v(\alpha_i) M_i\le_M v(\alpha_i)\frac{v(\alpha_i a) M_j}{v(\alpha_i)}=v(\alpha_i a) M_j
\end{equation*}
as required.
\end{proof}

\section{On the Necessity of WLP-axiom}\label{Example}
As we already mentioned, we do not know whether the WLP-axiom is necessary for the validity of the Theorem~\ref{th:Myhill-Nerode}.
However a non-uniform version of this axiom is always required if the monoid ${\cal M}$ is an mge and satisfies LSL- and GCLF-axioms.
When we say a non-uniform version of WLP-axiom we mean the following:
\begin{definition}\label{def:NUWLP}
For a monoid ${\cal M}$ and elements $u,v,x\in M$ and a sequence $\{a_n\}_{n=0}^{\infty}\subseteq M$ we define the predicate
$NUWLP(u,v,\{a_n\},x)$ as:
\begin{equation*}
NUWLP(u,v,\{a_n\},x) \iff [\forall n (ua_{n+1}=va_n)] \Rightarrow [ux\le_M vx \& \forall n (x\le_M a_n)].
\end{equation*}
We say that a monoid ${\cal M}$ satisfies the Non-Uniform Weak Limit Prefix Axiom (NUWLP-axiom) if:
\begin{equation*}
\forall u,v \forall \{a_n\}_{n=0}^{\infty} \exists x (NUWLP(u,v,\{a_n\},x)).
\end{equation*}
\end{definition}
\begin{remark}
Recall that the predicate $WLP(u,v,x)$ was defined as:
\begin{equation*}
WLP(u,v,x) \iff \forall \{a_n\}_{n=0}^{\infty} [\forall n (ua_{n+1}=va_n)] \Rightarrow [ux\le_M vx \& \forall n (x\le_M a_n)].
\end{equation*}
Thus, we can express $WLP(u,v,x)$ as $\forall \{a_n\}_{n=0}^{\infty}( NUWLP(u,v,\{a_n\},x))$. Consequently, we can
rewrite the definition of a WLP-axiom for a monoiid ${\cal M}$ as:
\begin{equation*}
\forall u,v \exists x\forall \{a_n\}_{n=0}^{\infty}  (NUWLP(u,v,\{a_n\},x)).
\end{equation*}
Comparing this formula with the definition of NUWLP-axiom:
\begin{equation*}
\forall u,v \forall \{a_n\}_{n=0}^{\infty} \exists x (NUWLP(u,v,\{a_n\},x)),
\end{equation*}
we see that the only difference is that in the WLP-axiom the witness $x$ depends only $u$ and $v$ but is uniform
for all the sequences $\{a_n\}$. On the other hand in NUWLP-axiom the witness $x$ depends besides on $u$ and $v$ also
on the specific sequence $\{a_n\}$. This explains the term we choose to name this axiom.
\end{remark}

\begin{lemma}\label{lemma:NUWLP-necessity}
Assume that ${\cal M}$ is an mge monoid such that every regular language $L\in Reg({\cal M})$ admits an infimum $\inf L\neq \emptyset$.

If further for every alphabet $\Sigma$ it holds that for every function $f:\Sigma^*\rightarrow {\cal M}$ with $ind(\equiv_f)\in \mathbb{N}$ there
is a subsequential transducer ${\cal T}$ with $f_{\cal T}=f$, then ${\cal M}$ satisfies the NUWLP-axiom.
\end{lemma}
\begin{proof}
Let $u,v\in M$ and $\{a_n\}_{n=0}^{\infty}\subseteq M$ be such that:
\begin{equation*}
u a_{n+1}=va_n \text{ for all } n\in \mathbb{N}.
\end{equation*}
Under the assumptions of the lemma, we have to show that there is some $x\in M$ such that $NUWLP(u,v,\{a_n\},x)$. 

To this end let us consider an alphabet $\Sigma=\{\sigma\}$ and the function $f:\Sigma^*\rightarrow M$ defined as:
\begin{equation*}
f(\sigma^n) = u a_n.
\end{equation*}
Let $s:\Sigma^*\rightarrow M$ be $s(\sigma^n)=a_n$. Then it is straightforward that $\tuple{u,v,s}$ is a witness for
$\varepsilon\equiv_f \sigma$. Since ${\cal M}$ is an mge monoid it follows that for $[\varepsilon]_{\equiv_f}=\Sigma^*$.
In particular, $ind(\equiv_f)=1$. By the assumptions of the lemma there is (sub)sequential transducer ${\cal T}$ with
$f_{\cal T}=f$.

Without loss of generality, and since $\Sigma=\{\sigma\}$ is a singleton, we can assume that there are some $j\le k$ such that:
\begin{eqnarray*}
{\cal T} &=& \tuple{\{\sigma\}\times M,Q,q_0,Q,\delta,\lambda,\iota,\Psi}\\
Q  &=& \{q_0,q_1,\dots,q_k\} \\
\delta(q_i,\sigma) &=& \begin{cases}
q_{i+1} \text{ for } i<k\\
q_j \text{ for } i=k.
\end{cases}\\
\lambda(q_i,\sigma)&=&m_i.
\end{eqnarray*}
Let $\Delta=\{\tuple{q_i,m_i,\delta(q_i,\sigma)}\,|\, i\le k\}$. Then for each $i\le k$ we can consider the automaton:
\begin{equation*}
{\cal A}_i = \tuple{{\cal M},Q,q_0,\{q_i\},\Delta,\iota,\Psi_i} \text{ where } \Psi_i(q_i)=\Psi(q_i).
\end{equation*}
Thus, by the Kleene Theorem, we have that $L_i={\cal L}({\cal A}_i)$ is regular and by the assumptions of the lemma
it admits an infimum $y_i \in \inf L_i$. Let $A_i=\{a_{i+l(k-j+1)}\,|\, l\in \mathbb{N}\}$ for $i\le k$. In particular, $A_i$ is not empty.
Then an easy computation shows that:
\begin{eqnarray*}
 L_i &=& u \circ A_i \text{ and }\\
 L_{i+1} &=& v\circ A_i \text{ for } i<k.
\end{eqnarray*}
Since ${\cal M}$ is an mge monoid, Lemma~\ref{lemma:associativity} implies that:
\begin{eqnarray*}
y_i \in \inf L_i &=& \inf (u\circ A_i) = u \inf A_i \\
y_{i+1} \in \inf L_{i+1} &=& \inf (v\circ A_{i}) = v \inf A_i.
\end{eqnarray*}
Since the left hand sides are well-defined, we conclude that $\inf A_i$ is not empty. Let us fix elements $x_i\in \inf A_i$ for $i\le k$. Then, we get:
\begin{eqnarray*}
y_i \sim_M u x_i \text{ for } i\le k\text{ and } 
y_{i+1} \sim_M v x_i \text{ for } i<k.
\end{eqnarray*}
Finally, we note that for $i=k$ we have that:
\begin{eqnarray*}
v\circ A_k&=& \{v a_{k+l(k-j+1)} \,|\, l\in \mathbb{N}\}\\
& =& \{u a_{k+1+l(k-j+1)}\,|\, l\in \mathbb{N}\} \\
& = & \{u a_{j+(l+1)(k-j+1)}\,|\, l\in \mathbb{N}\} \\
&\subseteq &\{u a_{j+l(k-j+1)}\,|\, l\in \mathbb{N}\}\\
&=& u \circ A_j.
\end{eqnarray*}
Since $v x_k \in \inf (v\circ A_k)$ because $x_k =\inf A_k$, and $ux_j$ is an infimum
for $u \circ A_j$ we conclude that $ux_j\le_M vx_k$. Since $\{x_i\,|\, i\le k\}$ is finite, it is also regular, and by the assumptions of the lemma, it admits an infimum $x\in \inf \{x_i\,|\, i\le k\}$. 

Finally, we prove that $NUWLP(u,v,\{a_n\},x)$. First:
\begin{equation*}
ux \le_M ux_{i+1} \sim_M v x_{i} \text{ for } i<k \text{ and } ux\le_M ux_j \le_M vx_k, 
\end{equation*}
we conclude $ux\le_M vx_i$ for each $i\le k$. Therefore $ux \le_M vx$, because $vx\in \inf \{vx_i\,|\, i\le k\}$. Furthermore, since each $a_n\in A_i$ for some $i$, we get that $x_i\le_M a_n$ and by transitivity, we get $x\le_M a_n$.
Therefore $NUWLP(u,v,\{a_n\},x)$.
\end{proof}
\begin{remark}\label{regular_natural}
Note that the only additional assumption in Lemma~\ref{lemma:NUWLP-necessity} is that the regular languages over ${\cal M}$ admit infimums.
On the other hand, to our best knowledge, all the results, up to Theorem~\ref{th:Myhill-Nerode} in Section~\ref{sec:MNR}, characterising the (sub)sequential rational functions in terms of congruence relations rely on this assumption. It is also natural to assume this property, in view of the early normal forms that is desirable.
\end{remark}

\begin{lemma}\label{lemma:NUWLP_WLP}
Assume that ${\cal M}$ is an mge monoid such that every regular language $L\in Reg({\cal M})$ admits an infimum $\inf L\neq \emptyset$.
Assume also that ${\cal M}$ obeys the NUWLP-axiom. 

If for elements $u,v\in M$ the set:
\begin{equation*}
W(u,v) = \{x\in M \,|\, \exists \{a_n\}_{n=0}^{\infty}(NUWLP(u,v,\{a_n\},x))\}
\end{equation*} 
is regular, then there is a witness $x_0\in M$ such that $WLP(u,v,x_0)$.
\end{lemma}
\begin{proof}
If $W(u,v)$ is regular, then by the assumptions of the lemma there is $x_0\in \inf W(u,v)$. We prove that $WLP(u,v,x_0)$.
Indeed, for each $x\in W(u,v)$ we have that $ux\le_M vx$. By the definition of $x_0$ we have that $ux_0\le_M ux\le_M vx$.
Thus, $ux_0\in low(vW(u,v))$. On the other hand, we have that:
\begin{equation*}
vx_0\in v\inf W(u,v) =\inf v W(u,v).
\end{equation*}
Therefore $ux_0\le_M vx_0$ by the definition of an infimum.

Finally, if $\{a_n\}_{n=0}^{\infty}$ is such that $ua_{n+1}=va_n$ then, by the NUWLP-axiom, there is a witness $x$ such that
$NUWLP(u,v,\{a_n\},x)$. Thus, $x\in W(u,v)$ and therefore $x_0\le_M x$. By $NUWLP(u,v,\{a_n\},x)$ we have that $x\le_M a_n$
for all $n$. Therefore $x_0\le_M x\le_M a_n$. This concludes the proof of the fact that $WLP(u,v,x_0)$.
\end{proof}

\begin{corollary}\label{cor:NUWLP_WLP}
Assume that ${\cal M}$ is an mge monoid such that every regular language $L\in Reg({\cal M})$ admits an infimum $\inf L\neq \emptyset$.
Assume also that ${\cal M}$ obeys the NUWLP-axiom. 

If for any elements $u,v\in M$ the set:
\begin{equation*}
W(u,v) = \{x\in M \,|\, \exists \{a_n\}_{n=0}^{\infty}(NUWLP(u,v,\{a_n\},x))\}
\end{equation*} 
is regular, the monoid ${\cal M}$ satisfies the WLP-axiom.
\end{corollary}
\begin{proof}
Immediate from the proof of Lemma~\ref{lemma:NUWLP_WLP} and the definition of the WLP-axiom.
\end{proof}
\begin{remark}
In view of Remark~\ref{regular_natural} and the result of Corollary~\ref{cor:NUWLP_WLP} the gap between the NUWLP-axiom and WLP-axiom seems to be not that big after all. We consider it challenging to (dis)prove the existence of an mge monoid, where every regular set admits an infimum, the monoid satisfies the NUWLP-axiom but does not satisfy the WLP-axiom.  
\end{remark}

\section{GCD Monoids and their Relation to MGE Monoids}\label{sec:gcd_mge}
In this section we compare another large class of monoids, the gcd monoids, with the class of monoids that we considered in the previous sections.
The gcd monoids are known to provide a characterisation for (sub)sequential rational functions in terms of congruence relations,~\cite{SouzaMasterThesis}. The basic notion in 
the gcd monoids is the \emph{division} $a | b$, which means that there is an element $c$ s.t. $ac =b$. In our notions this is exactly $a\le_M b$.
The notion of a $gcd(S)$ for a set $S$ coincides with our notion for $\inf S$. With these remarks we can restate the original definition of a gcd monoid,~\cite{SouzaMasterThesis},  as:
\begin{definition}\label{def:gcd_monoid}
A monoid ${\cal M}$ is called a gcd monoid if it satisfies the LC- and RC-axioms and for every $\emptyset\subsetneq S\subseteq M$, $\inf S\neq \emptyset$.
\end{definition}
\begin{remark}\label{rem:gcd_monoid}
In~\cite{SouzaMasterThesis,Sakarovitch09}, it has been shown that every transducer over a gcd ${\cal M}$ can be transformed in an equivalent
onward transducer. Further, for every function $f:\Sigma^*\rightarrow M$ where the monoid ${\cal M}$ is a gcd monoid it has been established that
the following are equivalent:
\begin{enumerate}
\item $ind(\equiv_f)=n$ for some $n\in \mathbb{N}$.
\item $f$ is (sub)sequential rational function.
\item there is a (minimal complete) subsequential transducer with $ind(\equiv_f)$ that represents $f$.
\end{enumerate}
\end{remark}
\begin{lemma}\label{lemma:mge_gcd_reduce}
Every gcd monoid ${\cal M}$ satisfies the RMGE-, LSL- and WLP-axioms. 
\end{lemma} 
\begin{proof}
First we establish that ${\cal M}$ satisfies the RMGE- and LSL-axioms. This is trivial for the LSL-axiom.
Let $a,b\in M$ and consider the set $S=\{a,b\}$. Since ${\cal M}$ is a gcd, $\inf S\neq \emptyset$. Hence, the LSL-axiom is valid.

Next, assume that $S=up(\{a,b\})\neq \emptyset$. Then, again since ${\cal M}$ is a gcd monoid, $\inf S\neq \emptyset$. Let $s\in \inf S$.
Since $a\in low(S)$, we have that $a\le_M s$. Similarly, since $b\in low(S)$, we have that $b\le_M s$. Therefore $s\in up(\{a,b\})$ and
consequently $s\in S$. This proves that $s\in \sup \{a,b\}$. Thus, the RMGE-axiom is valid.

So far we know that every gcd monoid is an mge monoid, since the gcd monoids satisfy the LC- and RC-axioms by definition.

Let $u,v\in M$. Let $A(u,v)$ be the set:
\begin{equation*}
A(u,v) = \bigcup\{\{a_n\}_{n=0}^{\infty} \,|\, \forall n (ua_{n+1}=va_n)\}.
\end{equation*} 
If $A(u,v)=\emptyset$, then any $x\in M$ witnesses for $WLP(u,v,x)$. Let $A(u,v)\neq \emptyset$ and let
$x_0=\inf A(u,v)$. Note that $v A(u,v)\subseteq u A(u,v)$. Indeed, for each $a\in A(u,v)$ we have that $va=ua'$ for some $a'\in A(u,v)$.
Now, since ${\cal M}$ is an mge monoid and $A(u,v)\neq \emptyset$, Lemma~\ref{lemma:associativity} implies that:
\begin{equation*}
u \inf A(u,v) = \inf uA(u,v) \text{ and } v\inf A(u,v) = \inf v A(u,v).
\end{equation*}
Therefore $ux_0\in \inf u A(u,v)$ and $vx_0\in \inf v A(u,v)$. Since $vA(u,v)\subseteq uA(u,v)$ we get $\inf uA(uv)\subseteq low(v A(u,v))$.
Therefore $ux_0\le_M vx_0$.

Finally, if $\{a_n\}_{n=0}^{\infty}$ is an arbitrary sequence such that $ua_{n+1}=v a_n$, we get that $a_n \in A(u,v)$ for each $n$ and consequently
$x_0\le_M a_n$. This proves that $WLP(u,v,x_0)$. Therefore ${\cal M}$ satisfies the WLP-axiom.
\end{proof}
\begin{remark}
As noted in~\cite{SouzaMasterThesis}, the tropical monoid restricted to the rational numbers, $\tuple{\mathbb{Q}^+_0,+,0}$ is not a gcd monoid.
However, it is obviously a sequentiable structure, and thus is an mge monoid with WLP-axiom (and also GCLF-axiom). This shows, that the mge monoids with WLP-axiom non-trivially extend the class of the gcd monoids. 
\end{remark}
\begin{remark}
However it should not be true that every gcd monoid satisfies the GCLF-axiom.
\end{remark}

\section{Conclusion}\label{sec:conclusion}
In this paper we provided a characterisation of (sub)sequential rational functions $f:\Sigma^* \rightarrow {\cal M}$
in terms of the congruence relation $\equiv_f$ for a large class of monoids. There two main issues that are not quite
satisfactory. First, it seems natural to consider monoids ${\cal M}$ where every regular language $L$ admits an infimum
$\inf L$. The GLCF-axiom guarantee this, but it is not necessary satisfied in every gcd monoid. On the other hand the gcd monoids
do not capture natural monoids and what is worse do not provide a constructive way to compute witnesses in $\inf L$ algorithmically.
The question is: Is there a finite set of first order formulae $\Phi$ over the language ${\cal L}=\tuple{e;\circ;=}$ where $e$ is a constant symbol.
$\circ$ is binary functional symbol, $=$ is the formal equality, such that:
\begin{enumerate}
\item gcd monoids model $\Phi$. 
\item for any monoid ${\cal M}$ modelling $\Phi$, every language $L\in Reg({\cal M})$ satisfies $\inf L\neq \emptyset$.
\item and whose constructive versions enable the algorithmic computation of an element in $\inf L$ for regular languages $L$ (given as automata, say).
\end{enumerate}
 
The second question is whether the premise for GCLF-axiom in Theorem~\ref{th:Myhill-Nerode} can be relaxed. Aesthetically, it would be much more 
satisfactory to have an assumption that the regular languages over ${\cal M}$ admit an infimum. Yet, the proof of Theorem~\ref{th:Myhill-Nerode} that 
we provided heavily relies on the GCLF-axiom in order to reduce the problem to finite sets. The main problem here is to gain a better insight in 
the structure of the range of the function $f:\Sigma^*\rightarrow {\cal M}$.  

Even if the answer of the above question might be not ultimate, it is still interesting to investigate the gap between the NUWLP- and WLP-axioms. 
More precisely, Corollary~\ref{cor:NUWLP_WLP} suggests the following question. Is there a monoid  ${\cal M}$ with the following properties:
\begin{enumerate}
\item ${\cal M}$ is an mge monoid,
\item every regular language $L\in Reg({\cal M})$ admits an infimum, i.e. $\inf L\neq \emptyset$.
\item ${\cal M}$ satisfies the NUWLP-axiom,
\item ${\cal M}$ violates the WLP-axiom.
\end{enumerate}
In view of Lemma~\ref{cor:NUWLP_WLP} the construction of such a monoid would be very delicate.
However, if such a monoid exists, the natural question would be, what the right balance between NUWLP- and WLP-axiom is so that we have a characterisation of (sub)sequential rational functions in terms of congruence relations.

\bibliographystyle{splncs03}
\bibliography{bibliography}

\end{document}